\documentclass[10pt,onecolumn]{svjour3}
\usepackage{latexsym}
\usepackage{color}
\usepackage{amsmath}
\usepackage{amsfonts,amssymb,bm}
\usepackage{array,cite,eqparbox,dcolumn}
\usepackage{graphicx}
\usepackage{subfigure}
\graphicspath{{figures/pdf/},{figures/png/}}
\DeclareGraphicsExtensions{.pdf,.png}
\smartqed
\def\diag{\operatorname{diag}}

\def\Sgn{\operatorname{Sgn}}

\def\ket#1{| #1 \rangle}
\def\bra#1{\langle #1 |}
\def\ip#1#2{\langle #1 | #2 \rangle}

\def\norm#1{\| #1 \|}

\def\up{\uparrow}
\def\down{\downarrow}
\def\S{\mathcal{S}}
\def\Da{\mathrm{Da}}
\def\prob{\mathrm{prob}}
\def\s{\mathsf{s}}

\makeatletter
  
  \spn@wtheorem{myexample}{Example}{\itshape\bf }{\rmfamily\slshape}
\makeatother
\newcommand{\ketbra}[3]{\langle #1 \vert #2 \vert #3 \rangle}
\newcommand{\inner}[2]{\langle #1 \vert #2 \rangle}

\renewcommand{\H}{\bar{H}}
\newcommand{\barcalH}{\bar{\mathcal{H}}}
\begin{document}

\title{Information Transfer Fidelity in Spin Networks and Ring-based Quantum Routers}

\author{E.~Jonckheere \and F. C. Langbein \and  S. G. Schirmer}
  \institute{E. Jonckheere \at Center for Quantum Information Science
            and Technology, University of Southern California,
            Los Angeles, CA, 90089, USA\\
            \email{jonckhee@usc.edu}
  \and
            F. C. Langbein \at School of Computer Science \& Informatics,
            Cardiff University, Cardiff CF24 3AA, UK\\
            \email{LangbeinFC@cardiff.ac.uk}
  \and
            S. G. Schirmer \at College of Science (Physics),
            Swansea University, Singleton Park,
            Swansea, SA2 8PP, UK\\
            \email{sgs29@swan.ac.uk}}
\date{}
\maketitle

\begin{abstract}
Spin networks are endowed with an Information Transfer Fidelity (ITF),
which defines an absolute upper bound on the probability of
transmission of an excitation from one spin to another.  The ITF is
easily computable but the bound can be reached asymptotically in time
only under certain conditions.  General conditions for attainability
of the bound are established and the process of achieving the maximum
transfer probability is given a dynamical model, the translation on
the torus.  The time to reach the maximum probability is estimated
using the simultaneous Diophantine approximation, implemented using
a variant of the Lenstra-Lenstra-Lov\'asz (LLL) algorithm.  For a
ring with uniform couplings, the network can be made a metric
space by defining a distance (satisfying the triangle inequality)
that quantifies the lack of transmission fidelity between two nodes.
It is shown that transfer fidelities and transfer times can be
improved significantly by means of simple controls taking the form
of non-dynamic, spatially localized bias fields, opening up the
possibility for intelligent design of spin networks and dynamic
routing of information encoded in them, while being more flexible than
engineering fixed couplings to favor some transfers, and less
demanding than control schemes requiring fast dynamic controls.

\end{abstract}

\section{Introduction}

Efficient and controllable transport of information is crucial for
information processing, both classical and quantum.  While bosonic
channels~\cite{Caruso2014} are the most attractive option for
long-distance communication, efficient on-chip interconnectivity in a
quantum processor based on atomic, ionic or quantum dot-based qubits,
or quantum spintronic devices~\cite{quantum_spintronics2013} will
require direct information transport through networks of coupled
solid-state qubits.  Such networks can be modeled via interacting
spins and are therefore generally referred to as spin networks.
Initiated Bose's~\cite{Bose2003} seminal work, spin networks have
received considerable attention in recent years (see review
articles~\cite{Bose2007,Key2010} and references therein).  Most of the
work has focused on information transmission through linear chains as
prototype quantum wires, starting with unmodulated
chains~\cite{Bose2003} and later perfect state transfer in chains with
fixed, engineered couplings~\cite{Christandl2004, Christandl2005}, and
finally controlled state transfer in spin chains, e.g., via adiabatic
passage ~\cite{Greentree2004}, ac modulation to achieve renomalization
of the couplings between adjacent qubits~\cite{Zueco2009}, single-node
bang-bang controls~\cite{Schirmer2009} or global dynamic
controls~\cite{Wang2010}.  Perfect state transfer in more general
networks has also been considered and some interesting results for
complete graphs were obtained in~\cite{Severini2009}.

Nonetheless, the information-theoretic properties of spin networks are
not fully understood.  Information encoded in excitations of a network
of coupled spins propagates, even under ideal conditions when quantum
coherence is maintained, in a non-classical way determined by the
Schr\"odinger equation.  Under best possible circumstances, this
propagation of excitations determines the Information Transfer
Fidelity (ITF) between various nodes of the network%
\footnote{Previously~\cite{MSC2011}, this concept was named
  Information Transfer \emph{Capacity}, but we refrain here from using
  this terminology to avoid confusion with the Shannon channel
  capacity~\cite{Caruso2014}.}.  Perfect state transfer between two
nodes can only be achieved when the Information Transfer Fidelity
between the respective nodes is unity.  However, this condition is not
sufficient. For example, while it is satisfied for the end nodes of a
chain with uniform couplings~\cite{MSC2011} such chains are usually
\emph{not} considered to admit perfect state transfer except
for chains of length two or three.

This raises the question of the attainability of the upper bound
given by the Information Transfer Fidelity.  Attainability in general
also depends on time constraints, i.e., attainable in what time, and the
margins of errors we are willing to accept.  In practice, some
margin of error is unavoidable, and the real question of interest is
therefore not whether we can achieve, e.g., \emph{perfect}, i.e.,
unit fidelity, state transfer in time $t_f$, but rather whether
we can achieve state transfer with a fidelity $1-\epsilon$, where
$\epsilon$ is an acceptable margin of error, in a reasonable amount
of time.  We may be willing to accept a slightly increased margin
of error for a significant reduction in the transfer time.
In this work, we are interested in such fundamental questions for spin networks
subject to coherent dynamics in general, and specifically simple
configurations such as a circular arrangement of spins (or spin ring
for short), which could serve as basic building blocks for more
complex architectures.

After introducing some basic definitions and basic results in Section
\ref{s:basic_def_results}, the concept of asymptotic ITF, i.e.,
maximum Information Transfer Fidelity attainable absent constraints on
the transfer times, between nodes in a network of spins is introduced
in Section \ref{s:attainability}.  Conditions for attainability of the
bounds are derived using dynamic flows on tori and the simultaneous
Diophantine approximation, computationally implemented using the
Lenstra-Lenstra-Lov\'asz (LLL) algorithm.  Under certain conditions,
the information transfer infidelity induces a metric that captures how
close two nodes in a spin network are from an information theoretic
point of view.  This information transfer geometry is investigated in
Section \ref{s:geometry}.  Finally, in Section~\ref{s:control}, we
investigate how the information transfer geometry of a network can be
changed by means of simple controls in the form of fixed biases
applied to individual nodes, and how this principle could be employed
for dynamic routing in a spin network with ring topology without the
requirement of fast-switching controls.

\section{Basic Definitions and Results}
\label{s:basic_def_results}

We consider networks of $N$ spins arranged in some regular pattern
with either XX or Heisenberg interaction~\cite{QINP2014} specified
by the Hamiltonian
\begin{align}
H = \sum_{i,j=1}^{N} J_{ij}
      \left(\sigma^x_i\sigma^x_{j} + \sigma^y_i\sigma^y_{j}
                             + \eta\sigma^z_i\sigma^z_{j} \right).
\end{align}
We specifically focus on networks with XX coupling ($\eta=0$) and
Heisenberg coupling ($\eta=1$), although most of the concepts and
analysis in the following are not limited to these types of coupling.
$J_{ij}$ is the strength of the coupling between spin $i$ and spin
$j$.  The factor $\sigma^{x,y,z}_i$ is the Pauli matrix along the
$x,y,$ or $z$ direction of spin $i$, i.e.,
\begin{equation*}
\sigma^{x,y,z}_i =
 I_{2\times 2} \otimes \ldots \otimes I_{2 \times 2} \otimes
 \sigma^{x,y,z}\otimes I_{2\times 2}\otimes \ldots \otimes I_{2 \times 2},
\end{equation*}
where the factor $\sigma^{x,y,z}$ occupies the $i$th position among the
$N$ factors and $\sigma^{x,y,z}$ is either of the single spin Pauli
operators
\begin{equation*}
\sigma^x= \begin{pmatrix} 0 & 1 \\ 1 & 0 \end{pmatrix}, \quad
\sigma^y= \begin{pmatrix} 0 & -\imath \\ \imath & 0  \end{pmatrix}, \quad
\sigma^z= \begin{pmatrix} 1 & 0  \\ 0 & -1 \end{pmatrix}.
\end{equation*}
The system Hilbert space $\mathcal{H}$ on which $H$ acts is conveniently taken as
$\mathbb{C}^{2^N}$.  We can abstract the network of spins as a graph
$\mathcal{G}=(\mathcal{V},\mathcal{E})$, where the vertices represent
the spins and the edges indicate the presence of couplings.

A particular configuration considered in this paper is that of \emph{spin
  rings}, i.e., spin networks defined by a circular arrangement of
spins, described by a $J$-coupling matrix that is \emph{circulant}
with \emph{nearest neighbor coupling}:
\begin{equation}
\label{e:Hamiltonian}
 H = \sum_{i=1}^{N-1} J_{i,i+1}
      \left(\sigma^x_i\sigma^x_{i+1} + \sigma^y_i\sigma^y_{i+1}
                             + \eta\sigma^z_i\sigma^z_{i+1} \right)
    +J_{N,1} \left(\sigma^x_N\sigma^x_{1}+\sigma^y_N\sigma^y_{1}
                                   +\eta\sigma^z_N\sigma^z_{1}\right).
\end{equation}
The term $J_{N,1}$ represents the coupling energy between the two
ends, spins $1$ and $N$, closing the ring. For networks with uniform
couplings, i.e., all non-zero couplings have equal strength $J$
(in units of Hz), we can set $J=1$ by choosing time in units of $J^{-1}$.

\subsection{Single Excitation Subspace}

Although many of the results in the following sections are more widely
applicable, we primarily concern ourselves here with the single
excitation subspace of the network~\cite{Key2010}, spanned by the $N$
single excitation quantum states $\{\ket{i}: i=1, \ldots, N\}$, where
$\ket{i}=\ket{\up\up\ldots \up\down\up\ldots\up}$ with $\down$ in the
$i$th position indicating that spin $i$ carries the excitation.  The
natural coupling among the spins allows the excitation at $i$ to drift
towards an excitation at $j$ with an \emph{Information Transfer
Fidelity (ITF)} that can be quantified by the maximum transition
probability $p_{\max}(i,j)$.  This concept will be precisely defined
in the next section, but in this introductory exposition we could
think of ``maximum'' as the process of giving the transition from
spin $i$ to $j$ the correct amount of time so that it is most
likely to occur.  The concepts behind these ideas are lying at the
foundation of quantum mechanics as embodied in the Feynman path
integral.  These concepts reveal that, contrary to classical
least-cost-path routing that follows a \emph{single} path from a
source to a destination in a classical network, quantum networks
follow \emph{all possible} paths from the state $\ket{i}$ to the state
$\ket{j}$.

\subsection{Eigendecomposition of the Hamiltonian}

Restricted to the single excitation subspace $\barcalH \cong
\mathbb{C}^N$, the eigen-decomposition of the Hamiltonian reads
$\bar{H}=\sum_k \lambda_k \Pi_k$, where $\lambda_k$ for
$k=1,\ldots,\widetilde{N}\le N$ are the distinct real eigenvalues and
$\Pi_k$ are the projectors onto the corresponding eigenspaces.

For a spin ring of size $N$ with uniform XX-coupling between adjacent
spins, $J_{ij}=J$ for $i=j\pm 1$, $(i,j)=(1,N)$, $(i,j)=(N,1)$, and $J_{ij}=0$
otherwise. In this case, the single excitation subspace Hamiltonian in the basis \linebreak
$\ket{i}=(0,0,\ldots,0,1,0,\ldots, 0)^T$ becomes the circulant matrix
\begin{equation}
\label{e:C_N}
  \H_N=\begin{pmatrix}
0 & 1 & 0 & \ldots &&&  0 & 1\\
1 & 0 & 1 &           &&&  0 & 0\\
0 & 1 & 0 & \ddots &&&  0 & 0\\
\vdots &  & \ddots  &&&  & &\vdots \\
  &   &   &        &&&  \ddots  & \\
0 & 0 & 0 &  &  &   \ddots  &  0  & 1\\
 1  & 0   &  0 & \ldots& & &   1  & 0
\end{pmatrix} (=: C_N),
\end{equation}
where the subscript $N$ is utilized to indicate that the system has
$N$ spins.  For uniform Heisenberg coupling the Hamiltonian is the
same except for the addition of a multiple of the identity, which
simply shifts the eigenvalues by a constant and does not affect the
eigenvector structure or differences between eigenvalues.  The
eigenvalues and eigenvectors of circulant matrices are well known and
shown in Table~\ref{t:eigenstructure}.  The $N$ single excitation
eigenvalues are conveniently parameterized by an integer $k$ running
from $0$ to $N-1$ or $1$ to $N$ with the cyclic condition that
$\lambda_{0}=\lambda_{N}$.

\begin{table}
\caption{Eigenvalues and eigenvectors of Hamiltonian $\H$ over single
  excitation subspace~\cite{Wang2012} in the basis where
  $\ket{i}=e_i:=\left(0 \ldots 0 ~ 1 ~ 0 \ldots 0 \right)^T$.
  $\rho_N=\exp(2\pi \imath /N)$ and $\ket{v_k}_j$ denotes the $j$th
  component of $\ket{v_k}$.}
\begin{center}
\begin{tabular}{|r|c|c|}\hline\hline
 & $\lambda_{k=0,...,N-1}$ & $\ket{v_k}_{j=0,...,N-1}$ \\ \hline\hline
&& \\
XX-coupling ($\eta=0$) & $2\cos\left( \frac{2\pi k}{N}\right)$ & $\sqrt{\frac{1}{N}} \rho_N^{k(j-1)}$ \\
&&\\ \hline
&&\\
Heisenberg coupling ($\eta=1$) &
$ 2\cos\left( \frac{2\pi k}{N}\right)+1$  & $ \sqrt{\frac{1}{N}} \rho_N^{k(j-1)} $ \\
&&\\\hline\hline
\end{tabular}
\end{center}
\label{t:eigenstructure}
\end{table}

The following lemma regarding the eigenvalues will be helpful later.
\begin{lemma} \label{l:interlacing}
For a spin ring of size $N$ with uniform XX-couplings we have:
\begin{itemize}
\item For $N$ even, but not divisible by 4, then the spectrum of
  $\H_N$ has mirror symmetry relative to the origin; precisely, we
  have $\lambda_k=\lambda_{N-k}= -\lambda_{N/2-k}=-\lambda_{N/2+k} \ne
  0$, i.e., there are $\tfrac{1}{2}N-1$ distinct pairs of double
  eigenvalues, and two single eigenvalues $\pm 2$, giving a total of
  $\widetilde{N}=(N+2)/2$ pairwise distinct eigenvalues:
      \[
      \{-2, \lambda_k, 2: k=1, \ldots, \tfrac{1}{2}N-1 \}.
      \]
      If $N$ is divisible by $4$ then the spectrum has a total of
      $\widetilde{N}=(N+2)/2$ pairwise distinct eigenvalues and a
      double eigenvalue at $0$ (for $k=\tfrac{1}{4}N, \tfrac{3}{4}N$).

\item For $N$ odd, we have $\lambda_{N-k} = \lambda_k \ne 0$ and there
  are $(N-1)/2$ distinct pairs of double eigenvalues and a single
  eigenvalue $+2$, giving a total of $\widetilde{N}:=(N+1)/2$ distinct
  eigenvalues:
      \[
          \{\lambda_k , +2: k=1,\ldots ,\tfrac{1}{2}(N-1)\}.
      \]
\item In either case, the number of pairwise distinct eigenvalues is
      \[
 \widetilde{N}:=\left\lceil \frac{N-1}{2}\right\rceil+1=\left\lceil\frac{N+1}{2}\right\rceil.
      \]
\end{itemize}
Moreover, the eigenvalues of $C_N$ and $C_{N-1}$ are interlaced.
\end{lemma}
\begin{proof}
The listed items are trivial. The last claim is the Cauchy interlacing
property~\cite{Massey2007}.
\end{proof}

For a double eigenvalue $\lambda_k=\lambda_{N-k}$, denote the
projection on the corresponding eigenspace as
$\Pi_k:=\ket{v_k}\bra{v_k}+\ket{v_{N-k}}\bra{v_{N-k}}$, where the
eigenvectors can be chosen such that $v_{N-k}=v_k^*$.  Moreover, for
the single eigenvalue $\lambda_0=+2$, define
$\Pi_0:=\ket{v_0}\bra{v_0}$ to be its eigenprojection.  If $N$ is
even, the single eigenvalue $\lambda_{N/2}=-2$ has its eigenprojection
denoted as $\Pi_{N/2}:=\ket{v_{N/2}}\bra{v_{N/2}}$.  If, in addition,
$N$ is divisible by $4$, denote the eigenprojection of the double
eigenvalue $\lambda_{N/4}=\lambda_{3N/4}=0$ as
$\Pi_{N/4}:=\ket{v_{N/4}}\bra{v_{N/4}}+\ket{v_{3N/4}}\bra{v_{3N/4} }$.
With this notation, the Hamiltonian restricted to the single
excitation subspace can be written as
\begin{equation*}
  \bar{H}=\sum_{k=0}^{\tilde{N}-1} \lambda_k \Pi_k.
\end{equation*}
  The
above can easily be extended to the Heisenberg case by globally
shifting the eigenvalues by $1$.

\section{Maximum Transfer Fidelity and Attainability}
\label{s:attainability}

Let $\ket{i}\in \barcalH$ be a quantum state with excitation localized
at spin $i$.  The quantum mechanical probability of transition from
state $\ket{i}$ to state $\ket{j}$ in an amount of time $t$ is given by
\begin{equation*}
  p_t\left(i,j\right)=|\ketbra{i}{e^{-\imath \hbar H_1t}}{j}|^2,
\end{equation*}
where we choose energies in units of $\hbar/J$ allowing us to assume
$\hbar=1$ and omit $\hbar$ in the following.  This formula is a
corollary of the Feynman path integral~\cite{Kaku1998, MIT2000}.  To
circumvent the difficulty posed by the time-dependence of this
probability, we proceed as in~\cite{MSC2011} and define the
\emph{maximum transition probability} $p_{\max}(i,j)$ also referred to
as \emph{Information Transfer Fidelity (ITF)}:
\begin{equation}
\label{e:ITF}
\begin{split}
  p_t(i,j)  & = \left| \bra{i} e^{-\imath H_1 t} \ket{j} \right|^2
                 =  \left| \sum_{k=0}^{\widetilde{N}-1} \bra{i}\Pi_k\ket{j}
                     e^{-\imath\lambda_k t}\right| ^2\\
             &\leq \left(\sum_{k=0}^{\widetilde{N}-1}
                       \left|\ bra{i}\Pi_k\ket{j} \right| \right)^2 =: p_{\max}(i,j).
\end{split}
\end{equation}
Clearly, $p_{\max}(i,j) \leq 1.$ Observe that, instead of taking the
sum of the absolute values of all $\bra{i}\Pi_k\ket{j}$ terms, we
could take the sum of the absolute values of some partial sums of such
terms and derive other upper bounds.  Note that the upper bound is
valid for any spin network, no matter how many spins, no matter how
many multiple eigenvalues, no matter the topology. Since the upper
bound depends only on the eigenvectors of the Hamiltonian and since
those are continuously dependent on the strengths of the couplings,
the upper bound is continuous relative to the $J_{ij}$.

\subsection{Attainability of Bounds}

The ITF $p_{\max}(i,j)$ is an upper bound on $p_t(i,j)$, which acquires
its full significance if $p_{\max}$ is achievable, that is, if there
exists a sequence of time samples $\{t_{i,j}(n): n \in \mathbb{N}\}$
such that $\lim_{n\to\infty} p_{t_{ij}(n)}(i,j) =p_{\max}(i,j)$.
Observing that the absolute value in Eq.~\eqref{e:ITF} will absorb any
global phase factor, the attainability condition is that there exists
$t\in [0,\infty)$ such that
\begin{equation}
\label{eq1}
   e^{-\imath\lambda_k t} = s_k(i,j) e^{\imath\phi}, \quad \forall k = 0, \dots,\tilde{N}-1,
\end{equation}
where $s_k(i,j) := \Sgn(\bra{i}| \Pi_k\ket{j}) \in \{0,\pm 1\}$ is a sign
factor and $\phi$ is a global phase, which is arbitrary but must be
the same for all $k$'s.  Eigenspaces with $s_k=0$ (where the $(i,j)$
dependency is suppressed to avoid the clutter) have no overlap with
the initial and/or target state and do not contribute to the sum.  We
shall refer to them as \emph{dark state} subspaces. They can be
ignored and we can restrict ourselves to the set
$K'\subseteq\{0,1,\ldots,\widetilde{N}-1 \}$ of indices $k$ for which
$s_k\neq 0$.  The physical interpretation of $K'$ is the set of
eigenspaces $\Pi_k \barcalH$ that have non-trivial overlap with the
initial and target state.  Noting that $s_k=\pm 1$ for $k\in K'$, and
$\exp[-\imath \tfrac{\pi}{2}(s_k-1)]=1$ for $s_k=1$ and $\exp[-\imath
  \tfrac{\pi}{2}(s_k-1)]=-1$ for $s_k=-1$, we can write
\begin{equation}
  s_k = \exp\left[-\imath\pi \left(2n_k + \tfrac{1}{2}(s_k-1)\right) \right],
\quad \forall k\in K',
\end{equation}
where $n_k \in \mathbb{Z}$ is an arbitrary integer.  Inserting this
into (\ref{eq1}), taking the logarithm and dividing by $-\imath$
yields
\begin{equation}
\label{eq2}
   \lambda_k t = 2\pi n_k + \tfrac{\pi}{2}(s_k-1)
   -\phi, \quad \forall k \in K'.
\end{equation}
This condition is not directly useful as $\phi$ can be arbitrary, but
we obtain meaningful constraints if we subtract the equations in a
pairwise manner, with $k \ne \ell$:
\begin{equation}
\label{eq3}
   (\lambda_k-\lambda_\ell) t =
2\pi (n_k-n_\ell) + \tfrac{\pi}{2}(s_k-s_\ell),
  \quad \forall k,\ell \in K'.
\end{equation}
We can also write the attainability constraints more explicitly:
\begin{align*}
   (\lambda_k-\lambda_\ell) t &= 2 \pi (n_k-n_\ell), \quad  &\mathrm{if} \quad
   & s_\ell=s_k,\\
   (\lambda_k-\lambda_\ell) t &= 2 \pi(n_k-n_\ell) + \pi, \quad &\mathrm{if} \quad
   & s_k=-s_\ell=1, \\
   (\lambda_k-\lambda_\ell) t &= 2 \pi(n_k-n_\ell) - \pi, \quad &\mathrm{if} \quad
   & s_k=-s_\ell=-1.
\end{align*}
These conditions are necessary and sufficient for attainability.  They
are physical, only involving differences of the eigenvalues, which are
observable and independent of arbitrary phases.  Vanishing left-hand
sides in the above are not an issue, as we are only looking at the
differences, which are non-zero by definition as $\lambda_k$, $k\in
K'$, are the distinct eigenvalues of $\bar{H}$.

Observe that all of the equations are compatible.  Indeed, adding
Eq.~\eqref{eq3} for $(k,\ell)$ and $(\ell,m)$ yields~\eqref{eq3} for
$(k,m)$.  Naturally, these equations are redundant, but we obtain a
set of linearly independent equations if we exclude the
dark state subspaces and restrict ourselves to a suitable subset of equations,
e.g., $(k_{i-1},k_i)$ or $(k_0,k_i)$ for $K' = \{k'_1,k'_2,\dotsc,K'_{\bar{N}}\}$.

\begin{myexample}[Dark States for Rings.]
For ring systems with uniform XX coupling, the distinct eigenvalues
are $\lambda_k=2\cos(2\pi k /N)$.  For eigenvalues of multiplicity
$1$, which occur for $k=0$, and $k=\tfrac{1}{2}N$ if $N$ is even,
$\bra{i}\Pi_0\ket{j}=(1/N) \ne 0$ and
$\bra{i}\Pi_{N/2}\ket{j}=(1/N)(-1)^{i-j} \ne 0$; therefore, there are
no dark states associated with these eigenvalues.  For eigenvalues
with multiplicity $2$, $\bra{i}\Pi_k\ket{j}
=\tfrac{2}{N}\cos(\tfrac{\pi}{2}n)$ with $n=4k(i-j)/N$ for
$k=0,\ldots,\lceil (N-4)/2 \rceil$; therefore, there are dark states if
and only if $n$ is an odd integer. This can happen only if $N$ is
divisible by $4$. The same holds of rings with uniform Heisenberg
coupling as they have the same eigenspace structure and the
differences between eigenvalues are the same.
\end{myexample}

\subsection{Simultaneous Attainability and Flows on the Torus}

Excluding dark state subspaces, restricting~\eqref{eq3} to a subset
$\S\subseteq K'\times K'$ of linearly independent equations, and setting
$\omega_{k,\ell}=(\lambda_k-\lambda_{\ell})/\pi$, the attainability
conditions become
\begin{equation}
\label{eq3prime}
  t\omega_{k, \ell}  = \tfrac{1}{2}(s_k-s_\ell) \mod 2, \quad
\omega_{k,\ell}:=(\lambda_k-\lambda_{\ell})/\pi, \quad
\forall (k,\ell) \in \S.
\end{equation}
The left-hand side of the above is the solution of the \emph{flow on
  the torus} $\dot{x}=\omega_{k\ell}$, with $x(0)=0$.  In this dynamic
formulation, the question is whether the flow starting at $x(0)=0$ passes
through the point with coordinates $0$ or $1$, depending on whether
$s_k=s_\ell$ or $s_k \ne s_\ell$, respectively.  It is
well-known~\cite[Prop. 1.5.1]{KatokHasselblatt1997} that the flow
starting at an \emph{arbitrary} $x(0)$ (which includes $x(0)=0$)
passes arbitrarily close to an arbitrary point on the torus if and
only if the $\omega_{k,\ell}$'s are linearly independent over the
rationals $\mathbb{Q}$.  This property of the flow getting arbitrarily
close to an arbitrary point from an arbitrary initial condition is
very strong and referred to as \emph{minimality}.  Observe that for
the flow to be minimal it suffices that, starting at $x(0)=0$, it gets
arbitrarily close to any point.  Obviously, minimality is sufficient
but not necessary for attainability, as the latter only requires the
flow to pass arbitrarily close to a specific point on the torus, while
minimality guarantees that the flow can get arbitrarily close to any
point.

Recall that Eq.~\eqref{eq3} refers to a specific but arbitrary
transfer $\ket{i}\to \ket{j}$, as the signs depend on $i,j$. We could
consider all Eq.~\eqref{eq3}'s for all $i \ne j$ and ask the question
as to whether there exists a unique $t$ such that attainability holds
for all $i \ne j$.  We refer to this stronger version of attainability
as \emph{simultaneous attainability}.

If for a given pair $(i,j)$ there are at least three non-dark
eigenspaces corresponding to $s_d,s_m,s_n \in \{\pm 1\}$, then there
must exist a pair, say $(m,n)$, with $s_m-s_n=0$.  In this case,
setting $t = 2\tau/\omega_{mn}$ for $\tau \in \mathbb{N}$ ensures that
the $(m,n)$ Eq~\eqref{eq3prime} holds \emph{exactly} and the remaining
attainability equations become
\begin{equation}
  \label{e:theta}
    \theta_{k\ell} \tau= \tfrac{1}{2}(s_k-s_\ell) \mod 2, \qquad
    \theta_{k\ell} := 2\omega_{k\ell}/\omega_{m,n}, \quad
    \forall\,(k,\ell)\in \S_0 := \S\setminus \{(m,n)\}.
\end{equation}
The left-hand side $\theta_{k\ell} \tau$ of the preceding equation is
the solution of the \emph{translation on the torus}, that is,
$x(\tau+1)=x(\tau)+\theta_{k\ell} \mod2$ with initial condition
$x(0)=0$.  By~\cite[Prop. 1.4.1]{KatokHasselblatt1997}, the
translation on the torus can come arbitrarily close to any point iff
the elements in the set $\{1\} \cup \{\theta_{k\ell} :
(k,\ell)\in\S_0\}$ are linearly independent over $\mathbb{Q}$.  As
before, the linear independence is sufficient, but not necessary for
attainability.  Note that we can in principle always reorder the
eigenvalues so that the reference transition is $(m,n)=(1,2)$.

It should be noted that the attainability criteria above apply to any
spin network.  For specific types of networks, we can derive more
explicit criteria.

\begin{myexample}[Attainability Condition for Rings.]
Given the formula for the eigenvalues for homogeneous rings,
$\lambda_k = 2\cos(2\pi k/N)$, elementary trigonometry shows that
\begin{equation}
  \omega_{k\ell} =  \tfrac{1}{\pi}(\lambda_k-\lambda_\ell)
  = -\tfrac{4}{\pi} \sin( \tfrac{\pi}{N} (k+\ell)) \sin(\tfrac{\pi}{N} (k-\ell)).
\end{equation}
There are $\widetilde{N}=\lfloor N/2\rfloor+1$ eigenspaces and
$\widetilde{N}-1$ independent transition frequencies $\omega_{k\ell}$.
Choosing the subset of linearly independent equations $\S = \{(k,k+1):
k=0,\ldots, \bar{N}\}$, $\bar{N}:=\widetilde{N}-2$, with the ordering
of the eigenspaces as defined above, the attainability conditions can
be written as
\begin{equation*}
   \tfrac{4}{\pi} \sin (\tfrac{\pi}{N}(2k+1))\sin (\tfrac{\pi}{N}) =
   \tfrac{1}{2}(s_k-s_{k+1}) \mod 2,  \quad \forall\, k=0,\ldots,\bar{N}.
\end{equation*}
If $s_m=s_{m+1}$, then setting $t=2\tau/\omega_{m,m+1}$ for $\tau \in
\mathbb{N}$ ensures $\omega_{m,m+1} t = 2\tau = 0 \mod 2$ and the
attainability conditions become
  \begin{equation*}
    \theta_k \tau = \tfrac{1}{2}(s_k-s_{k+1}) \mod 2, \quad \forall
    k=0,\ldots,\bar{N},
  \end{equation*}
with $\theta_k= \sin
(\tfrac{\pi}{N}(2k+1))/\sin(\tfrac{\pi}{N}(2m+1))$.  Notice that the
signs of the projections of the initial state $\ket{i}$ and target
state $\ket{j}$, $s_k = \bra{j}\Pi_k\ket{i}$, depend on the choices of
the latter, and it may happen that the signs $s_k$ are alternating,
$s_{k+1}=-s_k$ for all $k$.  In this case, the problem can easily be
rectified by reordering the eigenvalues, e.g., so that that $s_0' =
s_1'$ with the new ordering.
\end{myexample}

\begin{myexample}[Rational Independence.]
Applying the previous results to a ring of $N=5$ spins, the number of
pairwise distinct eigenvalues of the single excitation Hamiltonian
$\bar{H}$ is $\widetilde{N}=3$ and there are two linearly independent
transition frequencies
$\omega_{01}=-\tfrac{4}{\pi}\sin(\tfrac{1}{5}\pi)
\sin(\tfrac{1}{5}\pi)$ and $\omega_{12}=-\tfrac{4}{\pi}
\sin(\tfrac{3}{5}\pi) \sin(\pi \tfrac{1}{5})$.  To verify linear
independence of $\left\{\sin(\tfrac{1}{5}\pi),
\sin(\tfrac{3}{5}\pi)\right\}$ over $\mathbb{Q}$, we must show that
the equation
\begin{equation*}
     \alpha_1 \sin\left(\tfrac{ \pi}{5}\right)
  +\alpha_3 \sin\left(\tfrac{3\pi}{5}\right)  = 0 \text{ for }
  \alpha_1,\alpha_3 \in \mathbb{Q},
\end{equation*}
has only the trivial solution $\alpha_1=\alpha_3=0$ over
$\mathbb{Q}$.  Using $\sin\left(\frac{ \pi}{5}\right) = \frac{1}{4}
\sqrt{10-2\sqrt{5}}$ and $\sin \left(\frac{3\pi}{5}\right) =
\frac{1}{4} \sqrt{10+2\sqrt{5}}$ we can rewrite the equation as
$\alpha_1^2(10-2\sqrt{5}) = \alpha_3^2(10+2\sqrt{5})$.  Viewing the
field $\mathbb{Q}(\sqrt{5})$ as a two-dimensional vector space over
$\mathbb{Q}$ with basis $1,\sqrt{5}$ gives two equations
$\alpha_1^2=\alpha_3^2$ and $\alpha_1^2=-\alpha_3^2$, which much be
simultaneously satisfied.  This is possible only for $\alpha_1=\alpha_3=0$.  Thus
the flow on the torus is minimal and $p_{\max}(i,j)$ is attainable for
all $(i,j)$.
\label{ex:ring5}
\end{myexample}

\begin{myexample}[Rational Dependence for Even Rings.]
  For a ring with $N=10$ spins there are $\widetilde{N}=6$ distinct
  eigenvalues and five primary transition frequencies $\omega_{k,k+1}$
  for $k=0,\ldots,4$.  Noting that $\sin(\tfrac{5\pi}{10}) = 1$,
  $\sin\left(\tfrac{\pi}{10}\right) =
  \sin\left(\tfrac{9\pi}{10}\right) = \tfrac{1}{4}(-1+\sqrt{5})$ and
  $\sin\left(\tfrac{3\pi}{10}\right) =
  \sin\left(\tfrac{7\pi}{10}\right) = \tfrac{1}{4}( 1+\sqrt{5})$.  It
  is easily seen that $\alpha=(2,-2,1)$ is a $\mathbb{Q}$-solution to
  the linear dependence equation
  \begin{equation*}
    \alpha_1 \sin\left(\tfrac{9\pi}{10}\right)+
    \alpha_2 \sin\left(\tfrac{3\pi}{10}\right)+
    4 \alpha_3 \sin(\tfrac{5\pi}{10}) = 0.
   \end{equation*}
  Hence, the $p_{\max}(i,j)$ are not simultaneously attainable ---
  although $p_{\max}(i,j)$ may be attainable for some $(i,j)$.
\end{myexample}

More generally, for a ring with $N$ even there are $\tfrac{1}{2}N$ transition
frequencies
\begin{align*}
  \omega_{k,k+1} = \tfrac{4}{\pi} \sin \left((2k+1)\tfrac{\pi}{N}\right)
  \sin\left(\tfrac{\pi}{N}\right),
\end{align*}
which occur in pairs $\omega_{k,k+1}=\omega_{\bar{N}-k,\bar{N}-k+1}$
with $\bar{N}=\tfrac{1}{2}N-1$, precluding rational independence.

\begin{myexample}[Rational Dependence for Odd Rings.]
  Similarly, we can easily verify that for a ring with $N=9$ spins the
  transition frequencies are \emph{not} rationally independent, as we
  have, e.g. $\sin(7\pi/9)-\sin(5\pi/9) + \sin(\pi/9) = 0$ and thus
  $\omega_{3,4}-\omega_{2,3}+\omega_{0,1}=0$.
\end{myexample}

In general, rational independence of the transition frequencies for
homogeneous rings does \emph{not} hold when $N$ is not prime.

\subsection{Simultaneous Diophantine Approximation}
\label{s:simult_dio_approx}

Instead of checking rational independence of
$\{1\}\cup\{\theta_{k\ell}:(k,\ell)\in \S_0\}$, a less conservative
approach is to proceed, either analytically or
computationally~\cite{Lagarias1982,Kovacs2013}, via the
\emph{simultaneous Diophantine approximation}~\cite{Nowak1984,
  Hensley2005, Chevallier2011, Bosma2012} by finding
integers $p_{k\ell}, q$ such that
\begin{align*}
\left| \theta_{k\ell}-\frac{p_{k\ell}}{q} \right| \leq
\frac{c}{q^{1+\epsilon}}, \qquad  \forall (k,\ell) \in \S_0,
\end{align*}
and $\epsilon >0$.  With $\tau=q$, the above yields
\begin{equation*}
  \left| \theta_{k\ell}\tau-p_{k\ell} \right|
   \leq \frac{c}{\tau^\epsilon}, \qquad \forall (k,\ell) \in \S_0.
\end{equation*}

In the single-dimensional case, the solution is well-known to be given
by the continued fraction expansion of $\theta$.  Truncating the
continued fraction expansion yields convergents, i.e., rational
fractions $p/q$, with errors bounded as $|\theta q-p|\leq 1/q$, which
is optimal among all rational approximations of denominators less than
or equal to $q$.  The major hurdle at extending this result to the
multi-dimensional case is that there is an incompatibility between the
unimodular property of the Multi-dimensional Continued Fraction (MCF)
solution and optimality.

Nevertheless, the celebrated Dirichlet box principle shows that there
are multi-dimensional approximations with $c=1$ and $\epsilon =
1/\bar{N}$, where in the present context $\bar{N}=|\S_0|$.  Moreover,
there are \emph{infinitely many} integer solutions $q$ to the
simultaneous Diophantine approximation; in other words, as $\tau$ is
allowed to become arbitrarily large, the above error can be made
arbitrarily small.  The constant $c= 1$ can hardly be improved as for
$c<1$ there are \emph{``badly approximable vectors"} $\theta \in
\mathbb{R}^{\bar{N}}$ defined by
$\liminf_{q\to\infty}q^{1/\bar{N}}d(\theta q,\mathbb{Z}^{\bar{N}})>0$
such that the simultaneous Diophantine approximation has only
\emph{finitely} many solutions~\cite[Sec.\ 5]{Lagarias1982b,
  Hensley2005}.  If, however, $c$ is allowed to depend on $\bar{N}$,
refined bounds ($c<1$) can be derived on $c(\bar{N})$ due to the
existence of infinitely many solutions~\cite{Nowak1984}.  Specializing
the approximation to $\bar{N}=2$, it can be shown~\cite{Nowak1981}
that the bound can be improved down to $c=8/13$, along with
$\epsilon=1/2$. In the 1-dimensional case Hurwitz's theorem says that
one can take $c=1/\sqrt{5}$ and $\epsilon=1$.  On a general tone, the
Dirichlet approximation can only be improved slightly and at the
expense of considerable extra difficulties; we will therefore work
exclusively with the Dirichlet approximation in the following.

Assuming we have obtained a Dirichlet-good simultaneous Diophantine
approximation, the approximate attainability conditions become
\begin{equation}
  \label{e:specp}
  p_{k\ell} = \tfrac{1}{2} (s_k -s_\ell) \mod 2, \quad \forall (k,\ell) \in \S_0.
\end{equation}
The difficulty is to find, if it exists, a simultaneous Diophantine
approximation of Dirichlet accuracy that satisfies the above
conditions on the numerators.  The following example demonstrates that
it is not, in general, possible to achieve the even/odd
conditions~\eqref{e:specp} on the numerators $p_{k\ell}$ without
compromising on the accuracy of the Diophantine approximation.  To be
more specific, arbitrary accuracy can still be achieved with
Conditions~\eqref{e:specp}, but a larger denominator is required to
achieve the same level of accuracy in the presence of the constraints.

\begin{myexample}[Simultaneous Diophantine Approximation with
    Constraints.]
\label{ex:diophantine}
  In Example~\ref{ex:ring5} the flow on the torus
 for a ring with $N=5$ was found to be minimal, implying that we can
 get arbitrarily close to an arbitrary point on the torus.  By the
 preceding argument, this guarantees existence of simultaneous
 Diophantine approximations of arbitrary accuracies with prescribed
 even/odd numerators.  Furthermore, it is readily found that
\begin{equation*}
  \theta_{12} = \frac{2 \sin(3 \pi/5)}{\sin (\pi/5)}=1+\sqrt{5}
              = [3;4,4,4,4,4,...],
\end{equation*}
where the final expression denotes the continued fraction expansion
giving the \emph{optimal rational
  approximations}~\cite[Chap. 10]{Hua1982}.  It is known that
quadratic irrationality leads to continued fractions that eventually
stabilize.  The first convergents are
\begin{equation*}
 3, \frac{13}{4}, \frac{55}{17}, \frac{233}{72}, \frac{987}{305},
 \frac{4181}{1292}, \frac{17711}{5473}, \frac{75025}{23184},
 \frac{317811}{98209}, \frac{1346269}{416020},\ldots
\end{equation*}
Observe that \emph{all of them} have \emph{odd} numerators, while
the approximations we require must have \emph{even} numerators since
for the $N=5$ ring $s_0=s_1=1$.  This can be rectified by using the
so-called semi-convergents~\cite[Sec.\ V.4]{Koblitz1988},
\cite{Kovacs2013}. Given two convergents ordered as
$p_{n-1}/q_{n-1} < p_n/q_n$
one can easily squeeze a semi-convergent between them as follows:
\begin{equation*}
 \frac{p_{n-1}}{q_{n-1}}< \frac{p_{n-1}+p_n}{q_{n-1}+q_n}<\frac{p_n}{q_n}.
\end{equation*}
The semiconvergent has even numerator and has the accuracy of the
convergents $p_{n-1}/q_{n-1}$ and $p_n/q_n$ but at the cost of
doubling the denominator.
To prove that the semiconvergents provide approximations of arbitrary
accuracy, it suffices to show that there are infinitely many $n$'s
such that $p_{n-1}/q_{n-1} < p_n/q_n$. This is a corollary of the
unimodular property of continuous fractions, saying that
$p_{n-1}q_n-p_nq_{n-1}$ is alternately $\pm 1$.
\end{myexample}

We propose a general iterative method to deal with the even/odd
constraints.  To simplify the notation, let $\theta \in
\mathbb{R}^{\bar{N}}$, $p\in \mathbb{Z}^{\bar{N}}$, be
column-vectorizations of the $\theta_{k\ell}$s, $p_{k\ell}$s,
respectively, where $\bar{N}:=|\S_0|$.  We want to come up with a
Dirichlet-good approximation, $\theta \approx p/q$, where $p\in
\mathbb{Z}^{\bar{N}}$, $q \in \mathbb{N}$, with even/odd constraints
on the numerators $p_i$. By ``Dirichlet-good," we mean that the
infinity-error is bounded as $\| \theta q-p\|_\infty \leq
c/q^{1/\bar{N}}$, where $c$ is a constant independent of $\bar{N}$ and
$q$.  The idea is to iteratively scale $\theta$ by (the inverse of) a
diagonal matrix of positive rational numbers,
$\bar{\theta}=Y(n)^{-1}\theta$, compute a Dirichlet-good approximation
of $\bar{\theta}$ using, e.g., the Dirichlet box principle, or the
LLL-algorithm, or Lagarias' Multidimensional Continued Fractions
(MCFs), and then revise the scaling to meet the even/odd constraints,
with the hope that the procedure will converge.  Write the
Dirichlet-good approximation $\bar{\theta} \approx \bar{p}/q$ and
manipulate it as follows:
\[\begin{array}{ccccc}
  & & \|\bar{\theta} q - \bar{p} \|_\infty & \leq & \frac{1}{q^{1/\bar{N}}},\\
\min_i (Y(n)_{ii}^{-1})\|\theta q-Y(n)\bar{p}\|_\infty & \leq &
       \|Y(n)^{-1}(\theta q - Y(n)\bar{p}) \|_\infty & \leq & \frac{1}{q^{1/\bar{N}}}.
\end{array}\]
It follows that
\[ \|\theta q - Y(n)\bar{p}(Y(n))\|_\infty \leq \frac{1}{q^{1/\bar{N}}}\max_i Y(n)_{ii}. \]
%
In other words, $Y(n)\bar{p}/q$ is a Dirichlet-good approximation of
$\theta$ provided $\max_i Y(n)_{ii}$ can be dominated by a bound
independent of $q$ and $\bar{N}$.  Because the initial choice of
$Y(n)$ is arbitrary it is not guaranteed that $Y(n)\bar{p}$ has the
correct even/odd property. Nevertheless, we could revise $Y(n)$ to
meet those properties.  If a component $\bar{p}_i$ comes out to be odd
and needs to be even, we choose $Y(n+1)_{ii}=2$. If the algorithm has
converged, that is $Y(n+1)=Y(n)$, then the bound becomes
\begin{equation}
\label{e:full_generality}
\|\theta q - Y(n+1)\bar{p}(Y(n))\|_\infty \leq 2 \frac{1}{q^{1/\bar{N}}}.
\end{equation}
Conversely, if $\bar{p}_i$ comes out to be even with $2^d$ in its
prime number decomposition, we take $Y(n+1)_{ii}=1/2^d$ and, at
convergence, the bound on the $i$th component becomes
\[  |\theta_i q - 2^{-d}\bar{p}_i|\leq 2^{-d} \frac{1}{q^{1/\bar{N}}}. \]
Then this procedure is repeated with the scaling $\bar{\theta}=Y(n+1)
\theta$, in the hope that it converges.

\begin{theorem}
\label{t:even_odd_approximation}
Given $\theta \in \mathbb{R}^{\bar{N}}$, assuming $\{Y(n)\}$
converges, there exists a simultaneous Diophantine approximation
$\theta \approx p/q$ satisfying prescribed even/odd constraints on the
numerators $p_i$, $i=1,\ldots,\bar{N}$, with an error bound $\|\theta q -
p\|_\infty \leq 2/q^{1/\bar{N}}$ that is off the usual Dirichlet bound
by a factor not exceeding $2$.
\end{theorem}

\begin{myexample}[Simultaneous Diophantine Approximation with Constraints.]
We consider the same situation as in Example~\ref{ex:diophantine},
where all convergents of $ \theta_{12}$ have odd numerators while
attainability calls for an even numerator. We initiate the algorithm
with $Y(0)=1$, that is, $\bar{\theta}=\theta=1+\sqrt{5}$. Whatever
convergent $\bar{p}/q$ we pick, it has odd numerator, hence we take
$Y(1)=2$. We hence rewrite the continued fraction decomposition with
$\bar{\theta}=(1/2)\theta=(1+\sqrt{5})/2$, which gives the convergents
\[
    1, 2, \frac{3}{2},  \frac{5}{3},  \frac{8}{5},  \frac{13}{8},
    \frac{21}{13},  \frac{34}{21},  \frac{55}{34},  \frac{89}{55},
    \frac{144}{89},  \frac{233}{144}, \frac{377}{233},
    \frac{610}{377}, \frac{987}{610}, \ldots
\]
%
To secure convergence, $Y(3)=Y(2)=2$, we need to pick a convergent
with odd numerator, say, $377/233$, and the new Diophantine
approximation of $\theta_{12}$ is $2\times 377/233$.  This give an
error $|\theta_{12}\times 233-754|=0.0038<2/233=0.0086$, as claimed.
\end{myexample}

\subsection{(Weighted) LLL-Algorithm}

Even though Theorem~\ref{t:even_odd_approximation} guarantees that,
under convergence conditions, Dirichlet-good simultaneous Diophantine
approximations can be manipulated so as to yield numerators that have
prescribed even/odd properties, we are still left with the problem of
coming up with simultaneous Diophantine approximations in the first
place.

One of the first computational solutions to the simultaneous
Diophantine approximation was the so-called LLL-algorithm by Lenstra,
Lenstra and Lov\'asz~\cite{Lenstra1982, Chevallier2011, Kovacs2013}.
An alternative algorithm based on geodesic multi-dimensional continued
fraction expansion was proposed by Lagarias~\cite{Lagarias1994}.  Both
approaches proceed by reduction of the basis of the lattice generated
by the columns of
\begin{equation*}
  B(\mathsf{s})= \begin{pmatrix}
  I_{\bar{N}\times\bar{N}} & -\theta \\ 0_{1 \times \bar{N}} & \mathsf{s}
 \end{pmatrix},
\end{equation*}
where $\mathsf{s}\downarrow 0$ is a scaling parameter.  Observing that
$B(\s)(p,q)^T=(p-\theta q,\s q)^T$, it follows that a short vector in
the lattice $B(\s)\mathbb{Z}^{\bar{N}+1}$ yields a good approximation.
The numerator of this good approximation could be ``fixed" by the
procedure of Section~\ref{s:simult_dio_approx} to satisfy the even/odd
requirement.  However, it is proposed to combine the two procedures
into a single one---computation of a good approximation from a short
lattice vector and fixing the numerator---

by introducing a \emph{nonuniform} diagonal scaling
and work on the lattice $\Lambda(\s,X)$ generated by the columns of
\begin{equation*}
  B(\s,X)= \begin{pmatrix}
  X & -X\theta \\ 0_{1 \times \bar{N}} & \mathsf{s}
 \end{pmatrix},
\end{equation*}
where $X=\diag \left( x_1,\ldots, x_{\bar{N}}\right)$.  Note that for
$\s=1$ and $X=xI_{\bar{N}\times\bar{N}}$, we recover the scaling
of~\cite{Kovacs2013}.  Like the algorithm of
Section~\ref{s:simult_dio_approx}, this procedure is not guaranteed to
be successful, but if it is, it yields solutions guaranteed to be
optimal relative to some criterion.  The LLL-algorithm produces a
basis of short Euclidean norm vectors \linebreak $\left( b^*(\s,X)_1,
b^*(\s,X)_2,\ldots, b^*(\s,X)_{\bar{N}+1} \right)=:B^*(\s,X)$ such
that
\begin{equation*}
  \| b^*(\s,X)_1\| < \| b^*(\s,X)_j \|, \quad j=2,\ldots \bar{N}+1.
\end{equation*}
The $b^*(\s,X)_1$ vector is very close to the
shortest one.  A refined version of the LLL-algorithm captures the
genuinely shortest vector of the lattice $\Lambda(\s,X)$ as follows:
Given the reduced basis $\{b^*(\s,X)_i:i=1,\ldots,\bar{N}+1\}$, it
can be shown that the \emph{shortest} (in the sense of the Euclidean
norm) lattice vector is to be sought among all lattice vectors of the
form $\sum_i \beta_i b^*(\s,X)_i$, $\left|\beta_i\right| \leq
\left(2/\sqrt{3}\right)^{\bar{N}+1}$.  Lagarias' theorem~\cite[Lemma
  5]{Chevallier2011} then implies that a \emph{shortest} Euclidean
norm vector of the lattice is a \emph{best} $X$-weighted Diophantine
approximation.
Observing that
\begin{equation*}
B(\s,X) \begin{pmatrix} p                  \\ q   \end{pmatrix}
      = \begin{pmatrix} X(p-\theta q) \\ \s q \end{pmatrix},
\end{equation*}
and taking $\s \downarrow 0$, it becomes clear that a \emph{short}
vector in the lattice $B(\s,X)\mathbb{Z}^{\bar{N}+1}$ provides a
\emph{good} $X$-weighted Diophantine approximation:
\begin{equation}
\label{e:LLLapprox}
q=\frac{\left( B^*(\s,X)\right)_{\bar{N}+1,1}}{\s}, \quad
p_i=\frac{\left( B^*(\s,X)\right)_{i,1}}{x_i} + \theta_i q.
\end{equation}
With the \emph{shortest} vector, we construct the \emph{best}
approximation, that is, the approximation that minimizes
\[ \|X(\theta q - p )  \|_2, \]
 in the
same way as for the good approximation.

Before proceeding any further, we take care of a technicality:
As one would expect, the simulations also suggest that $q$ grows
without bound as $\s$ decreases to zero.  For the weighted
LLL-algorithm, we can prove the following:

\begin{theorem}
For the weighted LLL-algorithm to solve
\begin{equation}
\label{e:the_problem}
(\hat{p}(\s),\hat{q}(\s))
=\arg \min_{(p,q)\in \mathbb{Z}^{\bar{N}+1}}
\left\| B(\s,X) \begin{pmatrix} p\\ q \end{pmatrix} \right\|_{X \oplus 1}
\end{equation}
where $\| \cdot \|_{X \oplus 1}$ is the Euclidean norm
weighted by the direct sum of $X$ and $1$, we have $\lim_{\s \downarrow 0}
\hat{q}(\s)=\infty$.
\end{theorem}

\begin{proof}
Assume that there exist $\s_{\mathrm{min}}$ and $q_{\max}$ such that,
$\forall \s \leq \s_{\min}$, we have $q \leq q_{\max}$.  Consider
\eqref{e:the_problem} for any $0<\s\leq \s_{\min}$. By contradicting
hypothesis, $\hat{q}(\s) \leq q_{\max}$.  The above yields a
Diophantine approximation of $\theta$ but not the optimal one as $\s
\ne 0$. Now define
\begin{equation*}
  (\tilde{p},\tilde{q})=\arg \min_{(p,q)\in \mathbb{Z}^{\bar{N}+1}} \|p-\theta q\|_X
\end{equation*}
along with
\begin{equation*}
  \delta(\s)=\|\hat{p}(\s)-\theta\hat{q}(\s)\|_X^2-\|\tilde{p}-\theta\tilde{q} \|_X^2.
\end{equation*}
Observe that there exists a lower bound $\delta_{\min}$ such that
$\delta(\s) \geq \delta_{\min}>0$ as
$\|\hat{p}(\s)-\theta\hat{q}(\s)\|_X$ cannot reach its minimum since
$\hat{q}(\s) \leq q_{\max}$.  Now, consider the original
problem~\eqref{e:the_problem} with $\s<\min
\left\{\frac{\sqrt{\delta_{\min}}}{\sqrt{2}\tilde{q}},\s_{\min}\right\}$.
With this choice, we have
\begin{equation*}
  (\s\tilde{q})^2
  < \tfrac{\delta_{\mathrm{\mathrm{min}}}}{2}
  <\tfrac{\delta(\s)}{2}
  <\tfrac{\delta(\s)}{2}+(\s\hat{q}(\s))^2.
\end{equation*}
Then we have
\begin{align*}
  \left\| \begin{pmatrix}
           \tilde{p}-\theta \tilde{q}b \\
           \s\tilde{q}
           \end{pmatrix} \right\|_{\diag(X,1)}^2
  &=    \|\tilde{p}-\theta \tilde{q}\|_X^2+(\s\tilde{q})^2 \\
  &\leq \|\hat{p}(\s)-\theta \hat{q}(\s)\|_X^2 - \delta(\s) +
            (\s\hat{q}(\s))^2+\tfrac{\delta(\s)}{2}\\
  &= \|\hat{p}(\s)-\theta \hat{q}(\s)\|_X^2+(\s\hat{q}(\s))^2-\tfrac{\delta(\s)}{2}\\
  &=\left\|\begin{pmatrix}
                \hat{p}(\s)-\theta \hat{q}(\s) \\
                \s\hat{q}(\s)
                \end{pmatrix}\right\|_{X \oplus 1}^2 - \frac{\delta(\s)}{2}\\
  &<\left\|\begin{pmatrix}
                \hat{p}(\s)-\theta \hat{q}(\s) \\
                \s\hat{q}(\s)
                \end{pmatrix}\right\|_{X \oplus 1}^2
       -\frac{\delta_{\mathrm{min}}}{2}.
\end{align*}
The above is clearly a contradiction to the optimality of $(\hat{p}(\s),\hat{q}(\s))$.
\end{proof}

Note that the result \emph{appears} trivial from
Eq.~\eqref{e:LLLapprox} except that the behavior of the last component
of the first vector of the reduced basis has not yet been explored in
the weighted case.

Comparison between the weighted LLL-algorithm, $X(\theta q -p)$, and
the one of Section~\ref{s:simult_dio_approx}, $Y^{-1}(\theta q - Y
\bar{p})$, indicates that a good choice of the weighting might be
$X=Y^{-1}$.  This is only a guiding idea, as $X=Y^{-1}$ would mean
that $Y\bar{p}=p$, that is, $Y \times \mbox{Dirichlet numerator}
(Y^{-1} \theta)=\mbox{Dirichlet numerator}(\theta)$, which does not
hold exactly.

For practical computation of the time steps $\tau=q$, we must find
numerators $p_{k\ell}$ that fulfill the odd/even constraints using
the LLL-algorithm. The nonuniform variant introduced above makes it
simpler to find suitable parameters $X$ and $\s$, but a search is still
required.  To automate the search we use a standard genetic algorithm
to find weight vectors $X$ with a user-defined $\s$ that minimize the
number of parity constraint violations of the $p_{k\ell}$.  This
works well in most cases, requiring only a few iterations (typically
up to $5$) for reasonably sized populations (about $200$). We suggest
that the standard crossover and mutation operators could be adjusted
to improve the performance of the search.  In particular, increasing
the likelihood of changing the $X$ values corresponding to
denominators $p_{k\ell}$ that violate a constraint, and increasing
the likelihood of retaining $X$ values for which the corresponding
$p_{k\ell}$ do not violate the constraints may improve performance.

\begin{myexample}[Weighted LLL-Algorithm.]
  Consider a ring with $N=7$ spins and $p_{\max}(1,3)$.  There are
  four eigenspaces with projectors $\Pi_k$, and three rationally
  independent transition frequencies $\omega_{k,k+1} = (-0.7530,
  -1.6920, -1.3569)/\pi$.  Noting that $s_k = \Sgn
  \bra{3}\Pi_k\ket{1}$ for $k=0,\dotsc,3$ yielding $\vec{s} =
  (s_0,s_1,s_2,s_3) = (1,-1,-1,1)$, we choose $\omega_{12}$ as
  reference frequency and set
  \begin{equation}
    \vec{\theta}
    = 2 \omega_{12}^{-1}(\omega_{01},\omega_{23})^T
    = (0.8901, 1.6039)^T
  \end{equation}
  with corresponding constraints $\vec{s'} = (1,1)$, which means that
  the numerators $p_k$ in the simultaneous Diophantine approximation
  of $\vec{\theta}$ must both be odd.

  Applying the classical LLL-algorithm to solve the simultaneous
  Diophantine approximation for $\theta$ yields rational
  approximations of very high accuracy, as shown in
  Fig.~\ref{f:unweighted}.  However, most of the resulting
  approximations $p_k/q$ do not satisfy the parity constraints.  Using
  the weighted LLL-algorithm and varying the diagonal scaling vector $X$ enables
  us to find solutions of arbitrary accuracy, shown in
  Fig.~\ref{f:weighted}, all of which satisfy the parity constraints
  for the numerators $p_k$.

  For the approximation $p=(170921,307989)$ and $q=192028$ we obtain
  the transfer time $t_f = 2q/\omega_{12}=7.1308\times 10^5$ (in units
  of $J^{-1}$) and corresponding transfer fidelity
  \begin{equation}
    p_{t_f}(1,3) =  \left|\sum_{k=0}^4 e^{-\imath \lambda_k t_f} \bra{3}\Pi_k\ket{1}\right|^2
    \approx 0.4122,
  \end{equation}
  which is within $1-p_{t_f}(1,3)/p_{\max}(1,3) = 2.41\times 10^{-6}$
  of the maximum transfer fidelity $p_{\max}(1,3)$.
\end{myexample}

\begin{figure}[t]
\subfigure[Unweighted LLL-algorithm]
{\includegraphics[width=0.5\textwidth]{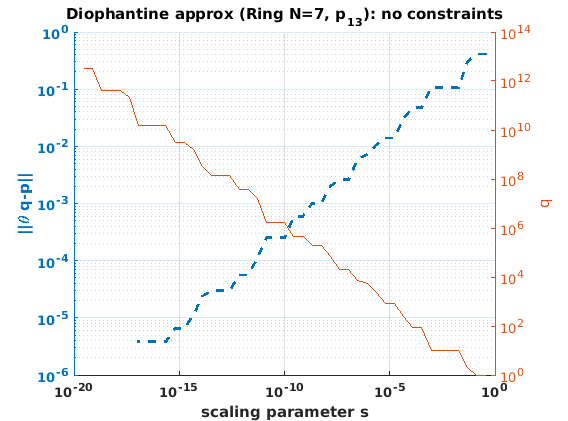}
\label{f:unweighted}}
\subfigure[Weighted LLL-algorithm]
{\includegraphics[width=0.5\textwidth]{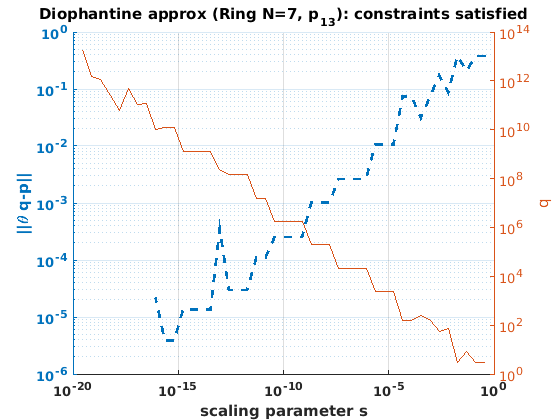}
\label{f:weighted}}
\caption{Behavior of LLL-algorithm applied to the simultaneous
  Diophantine approximation to determine attainability of
  $p_{\max}(1,3)$ in a $N=7$ ring.  The left vertical axis
  in both plots corresponds to the error
  of the approximation (thick broken line) while the right 
  vertical axis corresponds to the transfer time (thin solid line).}
\end{figure}

The previous example illustrates how we can use the weighted
LLL-algorithm to find optimal transfer times that yield very high transfer
fidelities, and how we can control the margins of error and ensure the
parity constraints are satisfied by adjusting the scaling parameter
and diagonal weights in the algorithm.

\subsection{Estimate of Time to Attain Maximum Probability}

Our objective is to find an upper bound on the amount of time $t$ it
takes to achieve $p_t(i,j) \geq p_{\max}(i,j) - \epsilon_{\prob}$,
i.e., $p_{\max}(i,j)-p_t(i,j) \leq \epsilon_{\prob}$.  The approach is
to translate the specification on the probability $\epsilon_{\prob}$
to a specification on the infinity-norm of the simultaneous
Diophantine approximation $\|\epsilon_{\Da}\|_{\infty}$, where
$\epsilon_{\Da}=\theta q-p$.

Proceeding from~\eqref{e:ITF}, recalling that $\Sgn\left(
\bra{i}\Pi_k\ket{j}\right) =: s_k =e^{-\imath \pi(s_k-1)/2-2\pi \imath
  n_k}$, where $n_k$ is some integer, we obtain
\begin{align*}
\sqrt{p_{\mathrm{max}}(i,j)}
=& \left| \sum_{k=0}^{\widetilde{N}-1}  \bra{i}Pi_k\ket{j} s_k  \right| \\
=&  \left| \sum_{k \in K'}  \bra{i}\Pi_k\ket{j} e^{-\imath\frac{\pi}{2}(s_k-1)-2\pi \imath n_k}  \right|\\
=& \left| \sum_{k\in K'}  \bra{i}\Pi_k\ket{j}
       e^{-\imath\frac{\pi}{2}(s_k-1)-2\pi \imath n_k} e^{\imath\frac{\pi}{2}(s_\ell-1)+2\pi \imath n_\ell} \right|\\
=& \left| \sum_{k\in K'}  \bra{i}\Pi_k\ket{j}
        e^{-\imath\frac{\pi}{2}(s_k-s_\ell)-2\pi \imath (n_k-n_\ell)} \right|.
\end{align*}
In the second equation, the sum over $k$ has been replaced by a sum
over $k \in K'$ as states with $\bra{i}\Pi_k\ket{j}=0$ do not contribute to
the sum.  The third equality stems from the fact that for fixed $\ell$,
$e^{\imath\frac{\pi}{2}(s_\ell-1)+2\pi \imath n_\ell}$ is a global
phase factor that is absorbed by the absolute value.

Next, we introduce the attainability condition~\eqref{eq3}, which is
only approximately satisfied using the simultaneous Diophantine
approximation.  The idea is to expose the gap between the left-hand
side and the right-hand side of~\eqref{eq3} when $t$ is constrained to
emerge from the Diophantine approximation:
\begin{align*}
   &\sqrt{p_{\max}(i,j)} \\
  = & \left| \sum_{k \in K'} \left(
           \bra{i}\Pi_k\ket{j} e^{-\imath(\lambda_k-\lambda_\ell)t}
       + \bra{i}\Pi_k\ket{j} \left( e^{- \imath \frac{\pi}{2} (s_k-s_\ell)-2\pi \imath(n_k-n_\ell)}
                                    -e^{-\imath (\lambda_k -\lambda_\ell)t} \right) \right) \right|  \\
\leq& \left| \sum_{k \in K'}
   \bra{i}\Pi_k\ket{j} e^{-\imath (\lambda_k-\lambda_{\ell})t} \right|
   +\left|\sum_{k\in K'} \bra{i}\Pi_k\ket{j}
             \left(e^{- \imath \frac{\pi}{2} (s_k-s_\ell)-2\pi \imath(n_k-n_\ell)}
                   -e^{-\imath (\lambda_k -\lambda_\ell)t} \right)\right| \\
=& \left| \sum_{k \in K'} \bra{i}\Pi_k\ket{j} e^{-\imath \lambda_k t} \right|
    +\left|\sum_{k\in K'} \bra{i}\Pi_k\ket{j} \left( e^{- \imath \frac{\pi}{2} (s_k-s_\ell)-2\pi \imath(n_k-n_\ell)}
                                                             -e^{-\imath (\lambda_k -\lambda_\ell)t} \right)\right| \\
\leq&~ \sqrt{p_t(i,j)} +\sum_{k\in K'}
        \left| e^{- \imath \frac{\pi}{2} (s_k-s_\ell)-2\pi \imath(n_k-n_\ell)} -e^{-\imath (\lambda_k -\lambda_\ell)t} \right| .
\end{align*}
It follows that $\sqrt{p_{\max}(i,j)}-\sqrt{p_t(i,j)} \leq
\sum_{k \in K'} \left| e^{- \imath \frac{\pi}{2} (s_k-s_\ell)} -
e^{-\imath (\lambda_k -\lambda_\ell)t} \right|$.  The trivial identity
\[
   p_{\max}-p_t = \left(\sqrt{p_{\max}}  -  \sqrt{p_t} \right)
                           \left(\sqrt{p_{\max}} + \sqrt{p_t} \right)
\]
then shows that to secure $p_{\max}(i,j)-p_t(i,j)\leq
\epsilon_{\prob}$, it suffices to require
\begin{equation}
\label{e:tspecs}
\sum_{k \in K'} \left| e^{- \imath \frac{\pi}{2} (s_k-s_\ell)}
  -e^{-\imath (\lambda_k -\lambda_\ell)t}  \right|
\leq \frac{\epsilon_{\prob}}{2},
\end{equation}
where it is observed that $\ell \in K'$ is arbitrary.

The last step is to relate the left-hand side of~\eqref{e:tspecs} to
the simultaneous Diophantine approximation error.  Define
\begin{equation}
  \epsilon_{\Da}(k,\ell):= \left|\theta_{k\ell}q -p_{k\ell}\right|, \qquad
  \norm{\epsilon_{\Da}}_\infty=\max_{(k,\ell) \in \S}\epsilon_{\Da} (k,\ell),
\end{equation}
where $\S$ is the subset of linearly independent attainability
equations chosen.  By definition any constraint $c_{k\ell} :=
\omega_{k\ell} t - \tfrac{1}{2}(s_k-s_\ell) = 0 \mod 2$ with
$\omega_{k\ell}=(\lambda_k-\lambda_\ell)/\pi$ can be written as a
linear combination of constraints with $(k',\ell') \in \S$, $c_{k\ell}
= \sum_{(k',\ell')\in\S} b_{k'\ell'} c_{k'\ell'}$, with coefficients
$b_{k'\ell'} \in \{0,\pm 1\}$.  Furthermore, given $\omega_{mn} \in
\S$ with $s_m=s_n$ and setting $t = 2\tau/\omega_{mn}$ with $\tau \in
\mathbb{N}$ and $\theta_{k\ell} = 2\omega_{k\ell}/\omega_{mn}$, we can
write the constraints as $c_{k\ell} = \theta_{k\ell} \tau -
\tfrac{1}{2}(s_k-s_\ell)$ for $(k,\ell) \in \S$.  Given a Diophantine
approximation that satisfies the parity constraints, $c_{k\ell}=
\epsilon_{\Da}(k,\ell) \mod 2$ for $(k,\ell)\in\S$, and $c_{k\ell}\le
\bar{N} \norm{\epsilon_{\Da}}_\infty$ for $(k,\ell)\not\in\S$, where
$\bar{N} = |\S|-1$ is the number of independent constraints reduced by
$1$.  Thus we have
\begin{align*}
\sum_{k \in K'} \left| e^{- \imath \frac{\pi}{2} (s_k-s_\ell)}-e^{-\imath \pi \omega_{k\ell}t} \right|
& =\sum_{k\in K'} \left| 1 - e^{- \imath \pi c_{k\ell}} \right| \\
&\le  |K'| \max_{k\in K'} \left| 1 - e^{- \imath \pi c_{k\ell}} \right| \\
&\le 2|K'|\left| \sin \left( \frac{\pi}{2} \bar{N} \|\epsilon_{\Da}\|_\infty \right) \right|.
\end{align*}
From the above string of inequalities, it follows that for the
attainability accuracy $\epsilon_{\rm prob}$ to be reached, it is
sufficient to take
\begin{equation}
\label{e:general}
2|K'|\left| \sin \left( \tfrac{\pi}{2} \bar{N}
  \|\epsilon_{\Da}(q)\|_\infty\right) \right|
  <\frac{\epsilon_{\mathrm{prob}}}{2}.
\end{equation}
We now summarize the situation we have reached:
\begin{theorem}
\label{t:time_estimate}
For homogeneous rings the ITF specification $p_t(i,j) \geq p_{\max}(i,j) -
\epsilon_{\mathrm{prob}}$ is achieved at time $t=2q/\omega_{mn}$ (in
$1/J$ units) if $q$ is chosen so that simultaneous Diophantine
approximation error $\epsilon_{\Da}(q) := \vec{p}-\vec{\theta} q$ has
its infinity norm satisfying~\eqref{e:general} and $\omega_{mn}$ is
the reference transition with respect to which
$\vec{\theta}=(\theta_{k\ell})$ was defined in \eqref{e:theta}.
\end{theorem}

There are many simultaneous Diophantine approximation schemes.  If we
retain the Dirichlet-good one with even/odd constraints on the
numerators, under the assumption that the algorithm of
Section~\ref{s:simult_dio_approx} converges, the error bound is
$\|\epsilon_{\Da}\|_\infty \leq 2/q^{1/\bar{N}}$, and we obtain the
further sufficient condition
\begin{equation*}
  2|K'|\left| \sin \left( 
    \frac{\pi\bar{N}}{q^{1/\bar{N}}}\right) \right|
    <\frac{\epsilon_{\prob}}{2}.
\end{equation*}
A minimum $q$ that guarantees $\epsilon_{\prob}$ is easily
extracted from the above inequality
\begin{equation}
\label{e:with_Dirichlet}
  q \geq \left(\frac{\pi
    \bar{N}}{\sin^{-1}\left(\frac{\epsilon_{\prob}}{4|K'|}\right)}\right)^{\bar{N}}
  \approx \left( \frac{4\pi\bar{N} |K'|}{ \epsilon_{\prob}} \right)^{\bar{N}},
\end{equation}
where the latter approximation uses $\sin(x)\approx x$ and is valid
if $x = \epsilon_{\prob}/(4|K'|)\ll 1$.

As an example will soon show, contrasting the above with numerical
simulations of Eq.~\eqref{e:ITF} reveals that the bound
$O\left(\bar{N}^{\bar{N}}\right)$ is very conservative, mainly because
the continuous-time dynamics on the torus was converted to a
discrete-time dynamics.  The conservativeness is somewhat mitigated by
the dimension reduction achieved by the elimination of dark states and
symmetries that reduce the number of relevant eigenspaces.  For
example, for homogeneous rings $\bar{N}\approx \tfrac{1}{2}N$ rather
than $N$.  Further improvement of the scaling behavior could be
achieved by utilizing tighter simultaneous Diophantine
approximations~\cite[Th. 2]{MSC2011},\cite{Nowak1984}, but at the
expense of significantly complicating the notation.  The reward for
the conservativeness of this bound is that it is \emph{quite general}
for rings with uniform coupling and their ITF attainable by the
algorithm of Section~\ref{s:simult_dio_approx}, as it depends on
neither the eigenvalues nor the odd/even pattern.  Furthermore, it
becomes \emph{very general} for any network subject to the mild
modification of replacing $\bar{N}$ by $N$ and $|K'|$ by $N$.

\begin{figure}
  \includegraphics[width=0.49\textwidth]{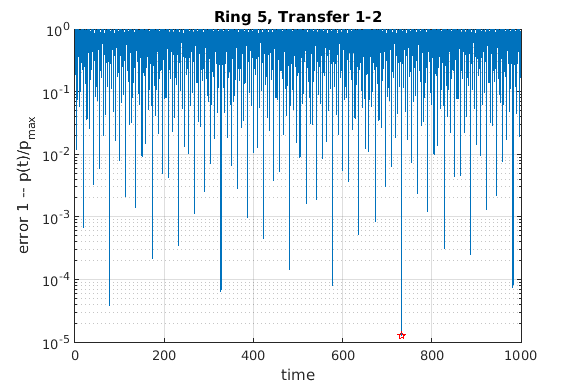}\hfill
  \includegraphics[width=0.49\textwidth]{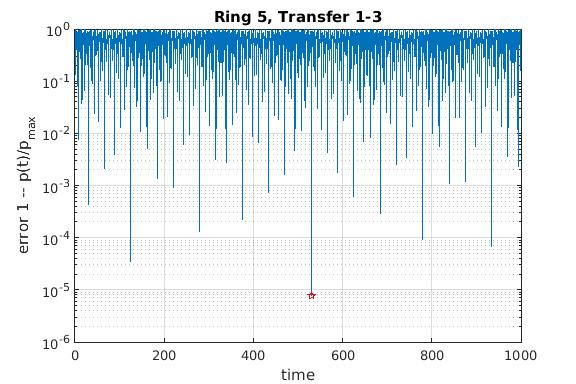}
  \caption{Simulations of transfer probabilities from $1\to 2$ and
    $1\to 3$ for a ring of size $N=5$.}\label{fig:ring5}
\end{figure}

\begin{myexample}[Transfer Times --- simulation results vs bounds.]
For a ring with $N=5$ we have $\widetilde{N}=3$ independent
eigenspaces with two rationally independent transition frequencies and
there are no dark subspaces.  Hence, $|K'|=3$ and we have $\bar{N}=1$
independent $\theta$.  In this case our conservative bound implies we
can get within $\epsilon_{\prob}$ of the maximum transition
probability in time $12\pi/\epsilon_{\prob}$.

In practice simulations suggest that we can we can achieve very high
fidelities in much shorter times. Fig.~\ref{fig:ring5} shows that we
can achieve $>99.99$\% of the maximum transfer fidelity for any two
nodes with distance $1$ in time $t=77.28$, and transfer between two
nodes with distance $2$ in time $t=125$ (in units of $1/J$).  Notice
that the maximum distance between any two nodes in a homogeneous ring
of size $N=5$ is $2$, hence any transfer can be achieved to within
$0.01$\% of the maximum possible in time $t\le 125$.

As observed before, for rings of size $N=6$, the primary transition
frequencies are \emph{not} rationally independent, implying that we do
not have simultaneous attainability.  Indeed, Fig.~\ref{fig:ring6}
(left) shows that the bound $p_{\max}(1,2)$ is not attainable.  Lack
of simultaneous attainability does not imply that all bounds are not
attainable.  Indeed.  Fig.~\ref{fig:ring6} (right) suggests near
perfect transfer between nodes of distance $n=2$.
\label{ex:transferTimes}
\end{myexample}

\begin{figure}
  \includegraphics[width=0.49\textwidth]{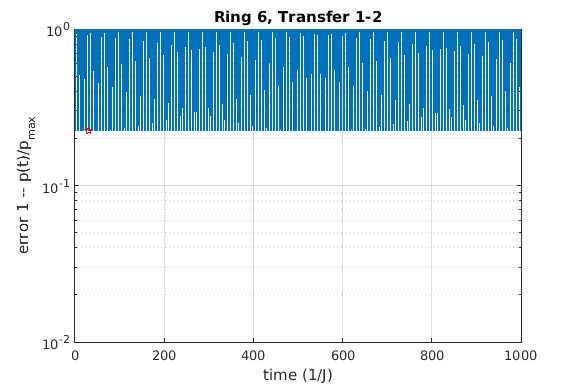}\hfill
  \includegraphics[width=0.49\textwidth]{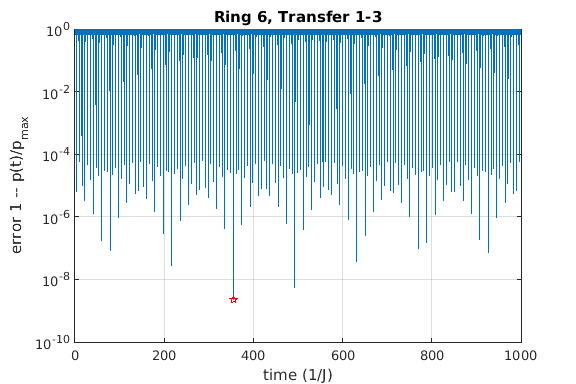}
  \caption{Simulations of transfer probabilities from $1\to 2$ and
    $1\to 3$ for a ring of size $N=6$.}\label{fig:ring6}
\end{figure}

We can use simulations combined with the LLL-algorithm to estimate the
minimum times required to achieve various transfers with a certain
maximum error probability.  The results for rings of size $N=5$ and
$N=7$, which satisfy the rational independence conditions for
simultaneous attainability, shown in Fig.~\ref{fig:scaling}, suggest a
power-law scaling.

\begin{figure}
  \includegraphics[width=0.49\textwidth]{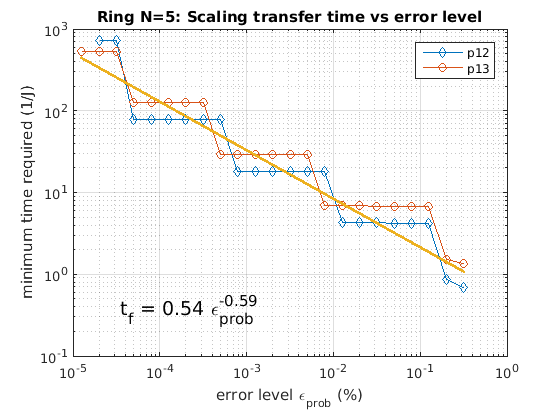}\hfill
  \includegraphics[width=0.49\textwidth]{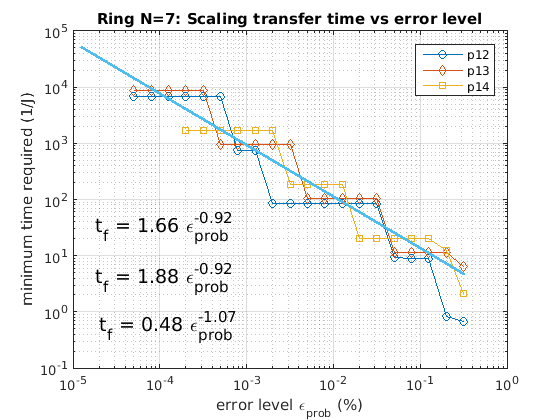}
  \caption{Transfer times estimated from simulations vs the error
    probability for ring of size 5 (left) and 7 (right).}
  \label{fig:scaling}
\end{figure}

Finally, comparing the scaling of the transfer times for rings of
different size in Fig.~\ref{fig:scaling2}(left) suggests that we
have similar scalings for both $N=5$ and $N=7$ although the constant
is larger for $N=7$.  The scaling behavior for various transfers
for a chain of size $N=7$ in Fig.~\ref{fig:scaling2}(right) is
similar but more complicated and the transfer times required to get
close to the upper bounds appear to be significantly longer.

\begin{figure}
  \includegraphics[width=0.49\textwidth]{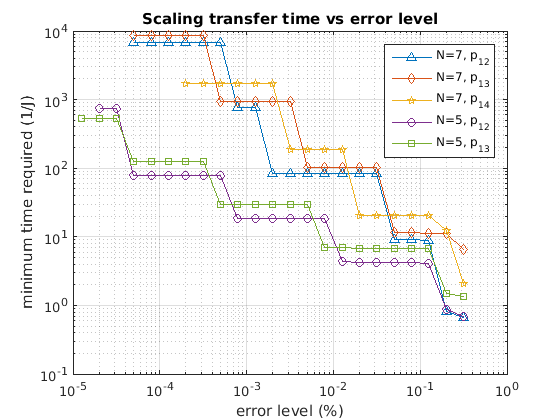}\hfill
  \includegraphics[width=0.49\textwidth]{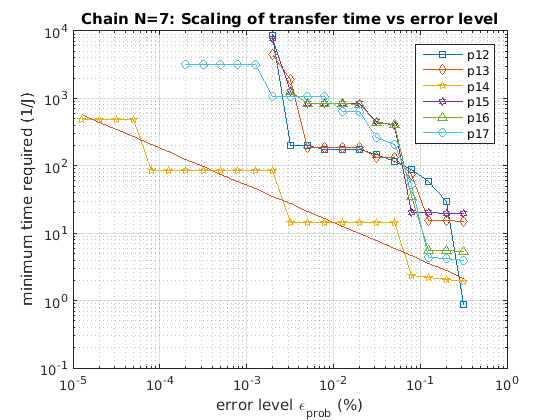}
  \caption{Comparison of scaling of transfer times with error
    probability for rings (left) and scaling for chains of size $N=7$
    (right).}
  \label{fig:scaling2}
\end{figure}

\subsection{Transfer Time versus Decoherence Time}

In general there is a trade-off between the error $\epsilon_{\prob}$
and the transfer time $t_f$ required to achieve
$p_{t_f}(i,j)=1-\epsilon_{\prob}$.  For actual physical realizations of
quantum networks, decoherence is generally a limiting factor.  In this
case the relationship between the error probability and the expected
transfer time can be useful in estimating what error probabilities can
be achieved based on the coherence time of the network
$t_{\mathrm{coh}}$.

For instance, in example~\ref{ex:transferTimes}, we showed that for a
ring of size $N=5$ we would achieve 99\% of the maximum transfer
probability between any two nodes in time $T\le125$ in units of the
inverse coupling rate $J^{-1}$, and we could therefore expect to
closely approximate these transfer fidelities provided the coherence
time of the system is $\gg 100 J^{-1}$.

More generally, Figures~\ref{fig:scaling}-\ref{fig:scaling2} suggest
that we have the power law $t_f=c\epsilon^{-\alpha}_{\prob}$ (in $1/J$
units) , at least for certain types of networks such as rings.  In this
case for the algorithm to work it is necessary that
\begin{equation}
  c\epsilon^{-\alpha}_{\prob} \ll t_\mathrm{coh}.
\end{equation}
This means that realistically, the error probabilities $\epsilon_{\prob}$
attainable are limited and we can expect
\begin{equation}
 \epsilon_{\prob}\gg (c/t_{\mathrm{coh}})^{1/\alpha},
\end{equation}
and the algorithm of Sec.~\ref{s:simult_dio_approx} could be used to
construct a simultaneous Diophantine approximation compatible with
this requirement.

Combining Th.~\ref{t:time_estimate} and Eq.~\eqref{e:with_Dirichlet}
also yields an upper bound on the transfer times for which the effect of
decoherence should definitively be negligible
\begin{equation}
t_f \leq \frac{2}{\omega_{mn}}
 \left( \frac{\pi \bar{N}} {\sin^{-1}\left( \frac{(c/t_{\mathrm{coh}})^{1/\alpha}}{4|K'|}\right)}  \right)^{\bar{N}},
\end{equation}
although we would like to stress here that this bound is excessively
conservative due to the approximations made.  Given a concrete
physical realization of a quantum network with an specific decoherence
model, this information could be used to derive tighter time-dependent
bounds on the transfer fidelities and realistic transfer times.

\section{Information Transfer (In-)fidelity Metric and Geometry}
\label{s:geometry}

In this section, we come back to an issue raised in
Section~\ref{s:attainability}---namely, that the upper bound derived
in Eq.~\eqref{e:ITF} can be justified by the fact that it induces a
metric on the set of vertices.  Unlike the results in the previous
sections, most of the results in this section apply specifically to
rings, although numerical simulations suggest that similar results may
hold for other homogeneous spin networks such as chains.

\subsection{Definition and Motivation of ITF Prametric}

To develop a geometric picture, we can view a spin network as a
pre-metric or more precisely a \emph{prametric space}%
\footnote{We
  prefer to avoid the terminology of \emph{pre-metric space} since
  it is not quite accepted; \emph{prametric} on the other hand is
  the terminology introduced by Arkhangel'skii and
  Pontryagin~\cite{Pontryagin1990}.}
endowed with the prametric that quantifies the \emph{Information
  Transfer Infidelity (ITI)}.  To fix terminology, recall that given a
graph $\mathcal{G}=(\mathcal{V},\mathcal{E})$, or any set of points
$\mathcal{V}$ for that matter, a
\emph{prametric}~\cite[p. 666]{Aldrovandi1995},
\cite[p.23]{Pontryagin1990} is a function $d: \mathcal{V} \times
\mathcal{V} \to \mathbb{R}_{\geq 0}$ such that \textbf{(i)} $d(i,j)
\geq 0$ and \textbf{(ii)} $d(i,i)=0$.

To derive a suitable prametric on the vertex set
$\mathcal{V}=\left\{\ket{i}:i=1,\dots,N\right\}$ from the probability
$p_{\mathrm{max}}$ data, we inspire ourselves from a similar situation
in sensor networks~\cite{EURASIP2008}, where $\mathcal{V}$ is the set
of sensors and a \emph{Packet Reception Rate} $\mathrm{PRR}(i,j)$ is
defined as the probability of successful transmission of the packets
from sensor $i$ to sensor $j$.  After symmetrization of the packet
reception rate, a prametric (in fact, a semi-metric~\cite{Wilson1931,
  Blumenthal1943, Shore1981}) can be defined as $d(i,j)=-\log
\mathrm{PRR}(i,j)$. Should there be a violation of the triangle
inequality, say, $d(i,j) > d(i,k)+d(k,j)$, then the distance between
$i$ and $j$ is redefined as $d(i,k)+d(k,j)$. The importance of the
metric is that it provides a notion of network curvature, which has a
dramatic impact on the traffic flow~\cite{Jonckheere2011, ACC2014} in
a paradigm that extends to quantum chains~\cite{QINP2014}.  Following
sensor network intuition\cite{EURASIP2008}, we define
\begin{equation}
  d(i,j) = -\log p_{\max}(i,j).
\end{equation}
Obviously, $d(i,j) \geq 0$ and, as will be shown in
Theorem~\ref{t:distance_properties}, $d(i,i)=0$.

We could define the time-stamped prametric by $d_t(i,j)=-\log
p_t\left(\ket{i},\ket{j}\right)$ except that in general $d_t(i,i) \ne
0$. To remedy this situation, we could define $d(i,j)=\inf_{t \geq 0}
d_t(i,j)=- \log \sup_{t \geq 0} p_t(i,j)$.  Since, by Cauchy-Schwarz,
$p_t(i,i) \leq 1$ and $p_{t=0}(i,i)=1$, we have $\sup_{t \geq
  0}p_{t}(i,i)=1$ and hence $d(i,i)=0$.  This alternate prametric
definition is equivalent to the earlier one when $p_{\mathrm{max}}$ is
attainable, but it reveals that this prametric makes the network of
finite diameter ($\sup_{i,j}d(i,j) < \infty$) as $N \to \infty$ as
Theorem~\ref{t:distance_properties} will show.  This has the
unfortunate consequence of preventing a genuine large-scale analysis.
As Section~\ref{s:control} will show a bias rectifies this problem
(see also~\cite{QINP2014}).

Generally, this information transfer infidelity prametric is not a
proper distance satisfying the triangle inequality, but for certain
networks such as rings with uniform coupling this prametric will be
shown to define a proper distance.

This quantum mechanical (pra)metric is quite different from the usual
Euclidean distance $d_{\mathbb{E}}$ of the spins in the spintronic
device.  In particular, two spins that are physically close in the
medium may be far quantum mechanically, and conversely.  If two spins
are quantum mechanically far, control is necessary to enable
transmissions that are too weak or forbidden by the natural quantum
mechanical couplings.  This control of information can be viewed as
the problem of controlling the quantum mechanical geometry of the
network.

\subsection{ITF Distance Geometry of Homogeneous Spin Rings}

It could be argued that a prametric is sufficient if we are solely
interested in assessing the difficulty of communication or fidelity of
information transfer between nodes in a network.  However, a proper
metric allows us to investigate other geometric properties such as the
curvature of the network with regard to the ITF.

A \emph{prametric} $d: \mathcal{V} \times \mathcal{V} \to
\mathbb{R}_{\geq 0}$ is a \emph{pseudo-metric} if in addition to {\bf
  (i)} $d(i,j)\ge 0$, $d(i,i)=0$, it satisfies {\bf (ii)}
$d(i,j)=d(j,i)$ and {\bf (iii)} the triangle inequality ($d(i,j) \leq
d(i,k)+d(k,j)$) holds.  A \emph{metric} or \emph{distance} is a
pseudo-metric that has {\bf (iv)} the separation property: $d(i,j)=0$ if
and only if $i=j$.

\begin{theorem}
  \label{t:distance_properties}
  For a quantum ring $(\mathcal{V}_N,\mathcal{E}_N)$ of $N$ uniformly
  distributed spins with XX or Heisenberg couplings, $d_N(i,j):=-\log
  p_{\max}(i,j)$ has the following properties:
\begin{enumerate}
\item For $N$ odd, $(\mathcal{V}_N,d_N)$ is a metric space.
\item For $N$ even, $(\mathcal{V}_N,d_N)$ is a pseudo-metric space that
      becomes metric after antipodal point identification.
\item If $N=p$ or $N=2p$, where $p$ is a prime number, then the
      distances on the space of equivalence classes of spins are
      uniform, i.e., $d_N(i,j)=c_N$ for $i\neq j$.  Otherwise, the
      distances are non-uniform.
\item In all cases $\lim_{N \to \infty} d_N(i,j)=2\log\tfrac{\pi}{2}$,
      $i\not= j \bmod (\tfrac{1}{2}N)$.
\end{enumerate}
\end{theorem}

\begin{proof}
To show that $(\mathcal{V}_N,d_N)$ is a pseudo-metric space we need to
verify that {\bf (i)} $d_N(i,i)=0$, {\bf (ii)} $d_N(i,j)=d_N(j,i)$, and {\bf (iii)} the
triangle inequality holds.  For a metric space we must further have
{\bf (iv)} $d_N(i,j)\neq 0$ unless $i=j$.

{\bf (i)} is clearly satisfied as the projectors onto the eigenspaces
are a resolution of the identity, $\sum_k\Pi_k = I$, and thus for any
unit vector $\ket{i}$, we have $\sum_{k=1}^N |\bra{i}\Pi_k\ket{i}|=
\sum_{k=1}^N \|\Pi_k \ket{i}\|^2 =1$.  {\bf (ii)} follows from
$|\bra{i}\Pi_k\ket{j}|=|\bra{j}\Pi_k\ket{i}|$.  The proof of the
remaining properties relies on the circulant matrix property of the
Hamiltonian $\H$ in the single excitation subspace $\mathcal{\H}$, as
shown in Eq.~\eqref{e:C_N} and Table~\ref{t:eigenstructure}.

Observe in Table~\ref{t:eigenstructure} the double eigenvalues
$\lambda_k=\lambda_{N-k}$, except for $k=0$ and $k=\tfrac{1}{2}N$ if
$N$ even.  From Table~\ref{t:eigenstructure}, each of these double
eigenvalues has two general complex conjugate eigenvectors.  These
general eigenvectors need not be orthogonal, but observing that
$\inner{v_k}{v_\ell}=\delta_{k\ell}$ and $\inner {v_k}{v_k^*}=0$,
where $v_k^*$ denotes the complex conjugate, it follows that
\begin{equation}
\begin{split}
  \ket{\bar{v}_0} &= \ket{v_0} = \tfrac{1}{\sqrt{N}}(1,1,\ldots)^T,\\
  \ket{\bar{v}_k} &= \ket{v_k}, \quad \ket{v_{N-k}} =\ket{v_k^*}, \quad k=1,\ldots N'=\lfloor\tfrac{N-1}{2} \rfloor, \\
  \ket{\bar{v}_{N/2}} &= \ket{v_{N/2}} = \tfrac{1}{\sqrt{N}}(1,-1,\ldots )^T, \quad \mbox{if }N\mbox{ is even},
\end{split}
\end{equation}
defines an orthonormal basis of $\mathcal{\H}$.  Furthermore, in the
basis in which $\H$ is circulant, we have $\ket{i}=e_i$, where
$\{e_i:i=1,...,N\}$ is the natural basis of $\mathbb{C}^N$.
\begin{align}
 |\bra{i}\Pi_0\ket{j}| &= |\inner{i}{\bar{v}_0}\inner{\bar{v}_0}{j}| = \tfrac{1}{N}, \label{e:harmonize1}\\
 |\bra{i}\Pi_k\ket{j}|
 &= |\inner{i}{\bar{v}_k}\inner{\bar{v}_k}{j} + \inner{i}{\bar{v}_{N-k}}\inner{\bar{v}_{N-k}}{j}| \nonumber\\
 &= \left|\rho_N^{ki} (\rho_N^{kj})^* +   (\rho_N^{ki})^* \rho_N^{kj}\right|\tfrac{1}{N}           \nonumber\\
 &= \left|\rho_N^{k(i-j)}+\rho_N^{-k(i-j)}\right|\tfrac{1}{N} = \tfrac{2}{N}\left|\cos(\tfrac{2\pi k(i-j)}N)\right|,  \label{e:harmonize2}\\
 |\bra{i}\Pi_{N/2}\ket{j}|
 &= | \inner{i}{\bar{v}_{N/2}}\inner{\bar{v}_{N/2}}{j}| = \tfrac{1}{N}          \label{e:harmonize3}.
\end{align}
Summing over all eigenspaces $k=0,\ldots,\lfloor N/2 \rfloor$ gives
\begin{align}
\label{eq:pmax}
 \sqrt{p_{\max}(i,j)}
=& \left\{ \begin{array}{ll}
    \frac{1}{N}+\frac{2}{N}\sum_{k=1}^{N'}\left|\cos\left(\tfrac{2\pi k(i-j)}{N}\right)\right|,
    & \qquad N=2N'+1, \\
    \frac{2}{N}+\frac{2}{N}\sum_{k=1}^{N'}\left|\cos\left(\tfrac{2\pi k(i-j)}{N}\right)\right|,
    & \qquad N=2N'+2.
\end{array} \right.
\end{align}
For $N=2N'+1$, it is easy to see that $p_{\max}(i,j)=1$ if and only if
$i=j$, hence {\bf (iv)}.  For $N=2N'+2$, on the other hand, we also
have $\left|\cos(\tfrac{2\pi k N/2}{N})\right|=|\cos(\pi k)|=1$, and
thus $d(i,j)=0$ for $i-j=\tfrac{1}{2}N$, i.e., the distance vanishes
for antipodal points, and thus $d(i,j)$ is at most a pseudo-metric.
However, noting that 
can identify antipodal points $\ket{j}$ and $\ket{j+N'+1}$, let $d$ be
defined on the set of equivalence classes $[\ket{j}]$ for $j=1,\ldots,
N'+1$ instead. (The antipodal identification preserves the ring
structure).  At this stage, $d$ is a
\emph{semi-metric}~\cite{Wilson1931, Shore1981,Blumenthal1953}, that
is, it satisfies all axioms of a metric except the triangle
inequality.

To prove the triangle inequality, we show that $\sqrt{p_{\max}(i,m)}
\sqrt{p_{\max}(m,j)} \le \sqrt{p_{\max}(i,j)}$.  The
definition~\eqref{e:ITF} of $p_{\max}$ rewritten in terms of the
eigenvectors of $\bar{H}$
using~\eqref{e:harmonize1}-\eqref{e:harmonize3} gives
\begin{align*}
  \sqrt{p_{\max}(i,m)}
  &= \frac{1}{N} \sum_{k=0}^{N-1}  \alpha_k  \rho_N^{k(m-i)}, \\
  \sqrt{p_{\max}(m,j)}
  &= \frac{1}{N} \sum_{k'=0}^{N-1} \beta_{k'}\rho_N^{k'(j-m)},
\end{align*}
where $\alpha_k=s_k(i,m)=\Sgn\left(\rho_N^{k(m-i)}+\rho_N^{-k(m-i)}
\right)\in\{\pm 1, 0\}$ is rewritten explicitly in terms of the
eigenvectors rather than as in Section~\ref{s:attainability} and
$\beta_{k'}=s_{k'}(m,j)$.  Setting
\begin{equation*}
  \gamma_k = \sum_{k'=0}^{N-1} \alpha_k \beta_{k'}\rho_N^{(k'-k)(j-m)}
\end{equation*}
we obtain
\begin{align*}
   & \sqrt{p_{\max}(i,m)} \sqrt{p_{\max}(m,j)} \\
 = & \frac{1}{N^2} \sum_{k,k'=0}^{N-1} \alpha_k \beta_{k'}
     \rho^{k(m-i)}_N \rho^{k'(j-m)}_N \\
 = & \frac{1}{N^2} \sum_{k,k'=0}^{N-1} \alpha_k \beta_{k'}
     \rho_N^{k(j-i) + (k'-k)(j-m)} \\
 = & \frac{1}{N^2} \sum_{k=0}^{N-1} \gamma_k \rho_N^{k(j-i)}
 = \left|\frac{1}{N^2} \sum_{k=0}^{N-1} \gamma_k \rho_N^{k(j-i)} \right|.
\end{align*}
The final equality follows because the LHS and thus the RHS are known
to be real and positive.  Furthermore, as $\rho_N$ is a root of unity,
$|\rho_N|=1$, and recalling $|\alpha_k|=|\beta_{k'}|=1,0$,
\begin{align*}
  |\gamma_k|
    &= \left|\rho_N^{k(m-j)} \sum_{k'=0}^{N-1} \alpha_k \beta_{k'}\rho_N^{k'(j-m)} \right| \\
    &\le \left|\rho_N^{k(m-j)}\right|\cdot \sum_{k'=0}^{N-1} \left|\alpha_k\beta_{k'}\rho_N^{k'(j-m)} \right|
   \leq N,
\end{align*}
where the last inequality allows for the presence of dark states.
Again we have $\rho_N^{(N-k)(m-j)}=\rho_N^{-k(m-j)}$, and as the LHS
above is known to be real, we know that we must have
$\gamma_{k}=\gamma_{N-k}$.  Hence, we can again collect exponential
terms pairwise to obtain cosines, which gives for $N=2N'+1$:
\begin{align*}
 \left|\frac{1}{N^2} \sum_{k=0}^{N-1}\gamma_k \rho_N^{k(j-i)}\right|
 &= \left|\frac{\gamma_0 }{N^2}
    + \frac{1}{N^2}\sum_{k=1}^{N'} 2 \gamma_k \cos\left(\tfrac{2\pi k(j-i)}{N}\right) \right|\\
 &\le \frac{|\gamma_0|}{N^2}
    + \frac{2}{N^2}\sum_{k=1}^{N'}|\gamma_k| \left|\cos\left(\tfrac{2\pi k(j-i)}{N}\right)\right| \\
 &\le \frac{1}{N} +
 \frac{2}{N}\sum_{k=1}^{N'}\left|\cos\left(\tfrac{2\pi k(j-i)}{N}\right)\right| \\
 &= \sqrt{p_{\max}(i,j)}.
\end{align*}
For $N=2N'+2$, we simply replace $\gamma_0$ by
$\gamma_0+\gamma_{N'+1}$ above to obtain
\begin{align*}
 \left|\frac{1}{N^2} \sum_{k=0}^{N-1}\gamma_k \rho_N^{k(j-i)}\right|
 & \le \frac{2}{N} +
 \frac{2}{N}\sum_{k=1}^{N'}\left|\cos\left(\tfrac{2\pi k(j-i)}{N}\right)\right| \\
 & = \sqrt{p_{\max}(i,j)}.
\end{align*}
This proves {\bf (iii)} and hence parts (1) and (2) of the theorem.

To establish (3), we note that if $N=2N'+1$ is prime then
\begin{equation*}
  \sum_{k=1}^{N'}\left|\cos\left(\tfrac{2\pi k(i-j)}{N}\right)\right|
 = \sum_{k=1}^{N'}\left|\cos\left(\tfrac{2\pi k}{N}\right)\right|.
\end{equation*}
If $N$ is not $p$ or $2p$ then $N$ and $(i-j)$ will have factors
(which can be canceled) in common for some $(i-j)$ but not for others
and hence we will obtain different distances.

To establish (4), letting $N \rightarrow \infty$, it is easily seen
that the dependency on $i,j$ is eliminated provided $i \not= j \bmod
(\tfrac{1}{2}N)$. Hence, taking the norm of the above and then $-\log(\cdot)$ it
follows that, at the infinite ring limit, the distance is uniform for $i \not=
j + \bmod (\tfrac{1}{2}N)$.  Finally,
\begin{align*}
\lim_{N\to \infty}\sqrt{p_{\mathrm{max}}(i,j)}
  &=\lim_{N\to\infty} \frac{2}{N} \sum_{k=0}^{N/2} |\cos((i-j)2\pi k/N)|\\
  &= \frac{2|i-j|}{\pi}\int_{0}^{\frac{\pi}{2|i-j|}} \cos(|i-j| x) dx \\
  &= \frac{2|i-j|}{\pi |i-j|} [\sin(|i-j| x)]_0^{\frac{\pi}{2|i-j|}}
   = \frac{2}{\pi}
\end{align*}
shows that $\lim_{N \to \infty}d_N(i,j)=2 \log \tfrac{\pi}{2}
\approx 2 \times 0.4516$ for $i \not= j \bmod(N/2)$.
\end{proof}

Case 3 of Theorem~\ref{t:distance_properties}  allows for a very specific
geometrization of the quantum ring in terms of constant curvature
spaces.  Define the $n$-sphere of curvature $\kappa$ as
$\mathbb{S}^{n}_\kappa:=\{x \in \mathbb{R}^{n+1}:\|x\|^2=1/\kappa \}$.
We have the following corollary:

\begin{corollary}
\label{c:sphere_embedding}
The metric space $(\mathcal{V}_p,d_p)$ of $p$ spins ($p\geq3 $ prime)
arranged in a homogeneous ring with uniform ITI distance
$d_p(i,j)=c_p$, $i \ne j$, is isometrically embeddable in
$\mathbb{S}^{p-1}_{\kappa}$ iff
\begin{equation}
\label{e:isometric_embedding}
\kappa \leq \left[ \frac{1}{c_p} \cos^{-1} \left( -\tfrac{1}{p-1} \right) \right]^2.
\end{equation}
%
%
Furthermore, it is \emph{irreducibly}
isometrically embeddable in $\mathbb{S}^{p-2}_{\kappa }$
for
\begin{equation}
\label{e:irreducible}
\kappa= \left[ \frac{1}{c_p} \cos^{-1} \left( -\tfrac{1}{p-1} \right) \right]^2.
\end{equation}
\end{corollary}
\emph{Notes:} In the above, \emph{``irreducibly} embeddable" means that
the embedding cannot happen into a lower-dimensional constant
curvature space.  By convention, $\cos^{-1}$ takes values in
$[\pi/2,\pi]$.
\begin{proof}
This result is a corollary of~\cite[Th. 63.1]{Blumenthal1953}. For the
details, see~\cite[Appendix]{ISCCSP2012}.
\end{proof}
Note that this corollary deals with embeddability \emph{of the vertices
  only;} however, edges can be mapped isometrically as arcs of great
circles on either the sphere of
curvature~\eqref{e:isometric_embedding} or that of
curvature~\eqref{e:irreducible}.  Also note that the symmetry of the
simple $p=3$ case of the circle $\mathbb{S}^1$ circumscribed to a
equilateral triangle is misleading, as in very high dimension ($p \to
\infty$), Eq.~\eqref{e:irreducible} yields $1/\sqrt{\kappa}=:R \to
\frac{c_p}{\pi/2}$, that is, all vertices are mapped to the
half-sphere of radius $R$.

Regarding $N=2p$ in Case 3, we could first do the
anti-podal identification on the \emph{combinatorial} ring
$(\mathcal{V}_{2p},\mathcal{E}_{2p})$, leading to a
$(\mathcal{V}_p,\mathcal{E}_p)$ ring, and then embed
$(\mathcal{V}_p,\mathcal{E}_p)$ as in the preceding corollary.

Regarding Case 4 when $N$ is odd, define $\epsilon:=\max_{i \ne
  j}|d_N(i,j)-2\log (\pi/2)|$.  Then the metric space
$(\mathcal{V}_N,d_N)$ can be mapped isometrically on the sphere
$\mathbb{S}^{N-2}_\kappa$ of radius
$d_\infty/\cos^{-1}\left(-(N-1)^{-1}\right)$ up to an additive
distortion not exceeding $\epsilon$, that is, the embedding is
quasi-isometric~\cite[7.2.G]{Gromov1987}.  The case of an even $N$ is
dealt with as before using anti-podal identification. The geometry of
a genuinely infinite ring ($N=\infty$ rather than $N \to \infty$) is
completely different and is left to future work.

The $N$ even case can be dealt with in a different way.  Rather than
doing, first, a combinatorial anti-podal identification ($i=j$ if
$i-j=0 \mod (\tfrac{1}{2}N)$) and, then, mapping the quotient space
$\mathcal{V}_N/\sim$ to the sphere, we could map the combinatorial
antipodal points to geometrical anti-podal points on the sphere
$\mathbb{S}^{N-2}_\kappa$ with the understanding that geometrical
antipodal points on the sphere are identified to yield the real
projective space $\mathbb{RP}^{N-2}$.  A slight generalization
of~\eqref{e:irreducible} of Corollary~\ref{c:sphere_embedding}
together with 4 of Theorem~\ref{t:distance_properties} yields an
irreducible embedding of $(\mathcal{V}_N,d_N)$ into the sphere of
curvature $\kappa=\left(\left(\cos^{-1}\left(-\frac{1}{N-1} \right)
\right)/\left(2\log\frac{\pi}{2}\right)\right)^2$.  On the other hand,
$\mathbb{RP}^{N-2}$ is usually endowed with the standard curvature $1$
metric of diameter $\pi/2$.  To sum up:

\begin{corollary}
For $N$ even, there is an embedding $\mathcal{V}_N \hookrightarrow
\mathbb{RP}^{N-2}$, which is quasi-isometric for the scaled distance
$d_N\cos^{-1}\left(-\frac{1}{N-1} \right)/\left(4 \log
\frac{\pi}{2}\right)$ on $\mathcal{V}_N$ and the curvature $1$
distance on $\mathbb{PR}^{N-2}$. Furthermore, for $N \to \infty$ the
distortion becomes vanishingly small.
\end{corollary}

\section{Control of Information Transfer Fidelity}
\label{s:control}

To overcome intrinsic limitations on quantum state transfer or speed
up transfer, one can either try to engineer spin chains or networks
with non-uniform couplings~\cite{Christandl2004,Christandl2005}, or
introduce dynamic control to change the network
topology~\cite{Zueco2009,Schirmer2009,Wang2010}.

Our analysis above shows that engineering the couplings is not
strictly necessary.  For an XX or Heisenberg-type chain with uniform
nearest-neighbor couplings, for example, it can easily be verified
that the information transfer fidelity between the end spins is unity,
and attainability of the bounds means that we can achieve arbitrarily
high state transfer fidelities between the end spins if we wait long
enough.  Engineering the couplings, however, can speed up certain
state transfer tasks such as state transfer between the end spins at
the expense of others.

A more flexible alternative to fixed engineered couplings is to apply
\emph{control} to change the network geometry and hence speed up
state transfer as well as enable some transfers that either were
forbidden or had poor ITF.  One way this can be achieved is to apply
static electromagnetic bias fields to change the energy-level
splittings between the spin-up and spin-down states for different
nodes in the graph, as suggested, e.g., in
  \cite{Severini2009}.  To see how the application of such bias
fields can alter the transfer fidelities and network geometry,
consider a simple, concrete example of a single bias field $\zeta$
applied to node $\ell$ in a spin ring with uniform coupling.  First,
due to translation invariance, we can always relabel the nodes so that
the biased node is node $N$.  Then, assuming XX coupling, the
Hamiltonian on the single excitation subspace becomes
\begin{equation}
\bar{H}^{(\zeta)}_N= \begin{pmatrix}
0 & 1 & \ldots  & 0 & 0 & 0 & \ldots & 0 & 1 \\
1 & 0 & \ldots  & 0 & 0 & 0 & \ldots & 0 & 0 \\
\vdots & \vdots & \ddots & \vdots & \vdots & \vdots   &    & \vdots   &\vdots     \\
0 & 0 & \ldots  & 0 &  1    & 0 & \ldots  & 0 & 0  \\
0 & 0 &\ldots   & 1 & 0 & 1 & \ldots  & 0 & 0  \\
0 & 0 & \ldots  & 0 & 1 & 0  & \ldots & 0 & 0\\
\vdots & \vdots &   & \vdots &   \vdots  & \vdots     & \ddots  & \vdots   & \vdots \\
0 & 0 & \ldots  & 0 &   0    & 0  & \ldots & 0        & 1          \\
1 & 0 & \ldots  & 0 & 0      & 0  & \ldots & 1        & \zeta
\end{pmatrix},
\label{e:perturbed_H_1}
\end{equation}
where it is observed that we have the decomposition
\begin{equation*}
 \H_N^{(\zeta)} = C_N + \zeta E_{N,N},
\end{equation*}
where $C_N$ is the $N \times N$ circulant matrix defined above and
$E_{N,N}$ is a $N\times N$ matrix which is zero except for a $1$ at
position $(N,N)$.

Physically, applying a large bias field to the $N$th node in the ring
results in a large detuning that effectively eliminates this node from
the ring and breaks the ring open, leaving a chain of length $N-1$.
Hence, in the limit $\zeta\to \infty$, we expect the transition
fidelities for the first $N-1$ nodes to approach those for a chain of
length $N-1$ while the transition fidelities between the first $N-1$
nodes and the final (biased) node approach $0$.  We now reformulate
this intuitively obvious result in precise mathematical language.

\begin{lemma}
\label{l:toeplitz}
The eigenvalues and eigenvectors of the $(N-1) \times (N-1)$ Toeplitz
matrix $T_{N-1}$ made up of ones on the super diagonal and subdiagonal
and zeros everywhere else are given by
$\lambda_k=2\cos\left(\tfrac{\pi k}{N}\right)$  and
$\ket{v_k}_i=\sqrt{\tfrac{2}{N}}\sin\left(\tfrac{\pi k i}{N}\right)$;
$k=1,\ldots,N-1$, $i=1,\ldots,N-1$.  Furthermore, for $k$ even,
$\ket{v_k}_1+\ket{v_k}_{N-1}=0$.
\end{lemma}

\begin{theorem}
\label{t:tricky_limit}
Let $p_{\rm chain}^{N-1}$ be the maximum transfer fidelities for a
spin chain of length $N-1$ with uniform coupling between nearest
neighbors.  Let $p_{\rm ring}^{N,\zeta}$ be the maximum transfer
fidelities for a ring of size $N$ with bias $\zeta$ on the $N$th node.
Then
\begin{equation}
\lim_{\zeta \to \infty}p_{\rm ring}^{N,\zeta}(i,j)=
 \begin{cases} p_{\rm chain}^{N-1}(i,j), & \mbox{if } i,j <N; \\
               0, & i=N,j \ne N \mbox{or } i \ne N, j=N;\\
               1, & i,j=N.
  \end{cases}
\end{equation}
\end{theorem}

\begin{proof}
Write the characteristic polynomial of $\bar{H}^{\zeta}_N$ as
$\det((\lambda I_N -C_N)-\zeta E_{N,N})$ and recall that the
determinant of the sum of two matrices equals the sums of the
determinants of all matrices made up with some columns of one matrix
and the complementary columns of the other matrix.  Applying the
latter to the characteristic polynomial of $\bar{H}^{(\zeta)}_N$
yields
\begin{align*}
  \det(\lambda I_N-\H^{(\zeta)})
 =\det(\lambda I_N - C_N) -\zeta \det\left(\lambda I_{N-1}-T_{N-1}\right),
\end{align*}
where $T_{N-1}$ is the Toeplitz matrix defined in the lemma.  From
classical root-locus techniques, it follows that, as $\zeta \to
\infty$, \emph{exactly one} eigenvalue $\lambda_N(\zeta)$ goes to
$\infty$, while the remaining ones $\lambda_1(\zeta),...,
\lambda_{N-1}(\zeta)$ converge to the roots of $\det\left(\lambda
I_{N-1}-T_{N-1}\right)=0$.

Next, we look at the eigenvectors and rewrite the eigenvector equation as
\begin{equation*}
\left(\begin{array}{c|c}
T_{N-1} &
\begin{array}{c}
1\\
0_{N-3} \\
1
\end{array}
\\ \hline
\begin{array}{ccc}
1 & 0_{N-3} & 1
\end{array}
& \zeta
\end{array}\right)
\left(\begin{array}{c}
\ket{v_k(\zeta)}_1\\
\vdots\\
\ket{v_k(\zeta)}_{N-1} \\ \hline
\ket{v_k(\zeta)}_N
\end{array}\right)=
\lambda_k(\zeta)
\left(\begin{array}{c}
\ket{v_k(\zeta)}_1\\
\vdots\\
\ket{v_k(\zeta)}_{N-1}\\ \hline
\ket{v_k(\zeta)}_N
\end{array}\right).
\end{equation*}
Consider first the first $k \ne N$ equations. Since $\lim_{\zeta \to
  \infty}\lambda_k(\zeta)$ exists and is finite, it follows from the
bottom eigenequation that $\zeta \ket{v_k(\zeta)}_N$ remains bounded
as $\zeta \to \infty$. Therefore, $\lim_{\zeta \to \infty}
\ket{v_k}_N=0$.  Since $\lambda_k(\infty)$ is a \emph{unique}
eigenvalue of $T_{N-1}$, it follows that $\lim_{\zeta \to
  \infty}\ket{v_k(\zeta)}_{1:N-1}$ is the corresponding eigenvector of
$T_{N-1}$.  It remains to show that with this $\ket{v_k}_{1:N-1}$ the
bottom eigenequation can be made to hold.  This is easily achieved by
defining
\begin{equation*}
  \lim_{\zeta \to \infty} \zeta \ket{v_k(\zeta)}_N
   = -\lim_{\zeta \to \infty} (\ket{v_k(\zeta)}_1+\ket{v_k(\zeta)}_{N-1})
\end{equation*}
By the lemma, for $k$ even, we have $\lim_{\zeta \to \infty} \zeta
\ket{v_k(\zeta)}_N =0$, and therefore the $k<N$ eigenequation holds with
$\ket{v_k(\zeta)}_N$ going to zero faster than $1/\zeta$.  For $k$ odd,
$\ket{v_k(\zeta)}_N$ goes to zero as $c/\zeta$, where $c \ne 0$ is some
constant.

By the root locus result, for $\zeta$ large enough, all eigenvalues
are distinct, and we have
\begin{align}
\label{e:almost_there}
\sqrt{p_{\mathrm{ring}}^{(\zeta,N)}(i,j)}
&= \sum_{k<N} |\ip{i}{v_k(\zeta)} \ip{v_k(\zeta)}{j}|
             + |\ip{i}{v_N(\zeta)}\ip{v_N(\zeta)}{j}| \nonumber \\
&= \sqrt{p_{\mathrm{chain}}^{(N-1)}(i,j)} + |\ip{i}{v_N(\zeta)}\ip{v_N(\zeta)}{j}|
\end{align}
where the second equality is understood as the $\zeta \to \infty$
limit.  To complete the proof, it therefore remains to look at
$\ket{v_N(\zeta)}$.

The last $k=N$ eigenequation easily implies that $\zeta
\ket{v_N(\zeta)}_{1:N-1}$ remains bounded as $\zeta \to
\infty$. Therefore $\lim_{\zeta \to \infty}
\ket{v_N(\zeta)}_{1:N-1}=0$.  To normalize the eigenvector, we take
$\lim_{\zeta \to \infty} \ket{v_N(\zeta)}_N=1$.  The latter together
with~\eqref{e:almost_there} proves the theorem.
\end{proof}

Thus we have a systematic way to compute the asymptotic transfer
probability of a ring with high bias from the transfer probability of
a chain without bias.

\begin{myexample}[Dynamic Routing.]
As an illustration of how these results can be used, consider a ring
of size $N=9$.  The maximum transfer fidelities between nodes $i\neq
j$ for this ring without bias are quite low, $0.4094$ and $0.4444$.
However, applying a large bias to node $9$ changes the maximum
transfer fidelities.  In particular, the maximum transfer fidelity
between nodes $1$ and $8$, $2$ and $7$, $3$ and $6$, and $4$ and $5$,
now approaches $1$.  Fig.~\ref{f:checker89} shows a visual
representation of the transfer fidelities for the ring without bias
(left) and with bias (right).  This result is consistent with
Theorem~\ref{t:tricky_limit}, as using Lemma~\ref{l:toeplitz}, it is
easily verified that
\begin{equation}
   \sqrt{p_{\mathrm{ring}}^{(8)}(i,9-i)}
= \frac{2}{9}\sum_{k=1}^8 \left|\sin\left(\frac{\pi i k}{9}   \right)
                                                      \sin\left(\frac{\pi(9-i)k}{9}\right)\right| =1.
\end{equation}
The example also shows that a finite bias is sufficient to enable
almost perfect state transfer in practice, despite the fact that the
ring only becomes a chain in the limit when an infinite bias is
applied to node $9$.  We also used the LLL-inspired algorithm to
estimate the transfer time as a function of the infidelity of the
transfer.  We note here that it was crucial to use the weighted LLL
algorithm to generate a range of simultaneous Diophantine
approximations, which generally did \emph{not} satisfy the parity
constraints on the numerators, and to use the idea of combining
approximations to satisfy the constraints.  With this approach we were
able to find solutions satisfying all of the parity constraints on the
numerators over a wide range of infidelities to estimate the transfer
times required as a function of the tolerated infidelity.  The
results, shown in Fig.~\ref{f:ring9-bias10-scaling} (left) suggest
that high fidelities are indeed attainable for modest biases, and the
apparent linearity of the data in the bilogarithmic plot still
suggests a polynomial scaling.  However, the actual transfer times are
significantly higher in this case than in previous examples.  We point
out here that our algorithm is not guaranteed to find the shortest
possible time although Fig.~\ref{f:ring9-bias10-scaling} (right)
suggests that there is a good correlation between the Diophantine
approximation error and the observed infidelity of the transfer.
Furthermore, the algorithm enables us to estimate necessary transfer
times far beyond the regime accessible by brute-force numerical
simulations.
\end{myexample}

This example shows how a dynamic routing scheme can be implemented to
transfer information from any node in a ring to any other node with
fidelity approaching unity by simply applying bias fields to different
nodes.  For transfer between nodes $1$ and $8$, $2$ and $7$, $3$ and
$6$, or $4$ and $5$, it suffices to apply a large bias to node $9$.
If we wish to transfer information from node $1$ to $4$ then
translation invariance of the ring allows us to shift the labels by
$2$, so that node $1$ becomes $3$ and $4$ becomes $6$, and applying a
bias to the new node $9$ will enable the transfer.

Further reflection shows that we can achieve maximum transfer
fidelities approaching unity for transfer between any pair of nodes in
a ring of size $N$ provided $N$ is odd by simply biasing the node in
the middle between the pair of spins. This is because in this case
$N-2$ is odd, so there must be an odd number of spins along one path
around the ring and an even number between the spins around the
other. By applying the bias in the middle of the path with an odd
number of spins we asymptotically reach a chain with $N-1$ (even)
spins. In this chain the transfer probability between spins mirrored
at the center is $1$, which is specifically true for the source and
target spin with an even number of spins between them in the chain.

If $N$ is even instead, then the situation is more complicated. If
there is an odd number of spins between source and target along the
ring, then applying a bias at the middle creates an odd chain where
source and target are connected with probability $1$ as they are at
mirror-symmetric positions in the ring. If there is an even number of
spins between source and target, then applying a single bias cannot
achieve perfect information transfer as the spins can never be at
mirrored positions in the odd chain (which are the only ones in the
chain perfectly connected). There are, however, multiple solutions to
apply a bias at two spins that can asymptotically generate a suitable
chain.

In practice it may be possible and even preferable to simultaneously
apply biases to several nodes instead of a single node to shape the
overall potential landscape.  This case is more difficult to treat
analytically but preliminary results~\cite{CDC2015} suggest that
numerical optimization can be used in this case to optimize the
applied biases to achieve significant reductions in the transfer times
and the magnitude of the required bias fields, as well as to deal with
practical issues such as leakage of the bias fields, i.e., the tendency
of a bias applied to one node to also affect nearby nodes.

\begin{figure}[t]
\includegraphics[width=0.49\textwidth]{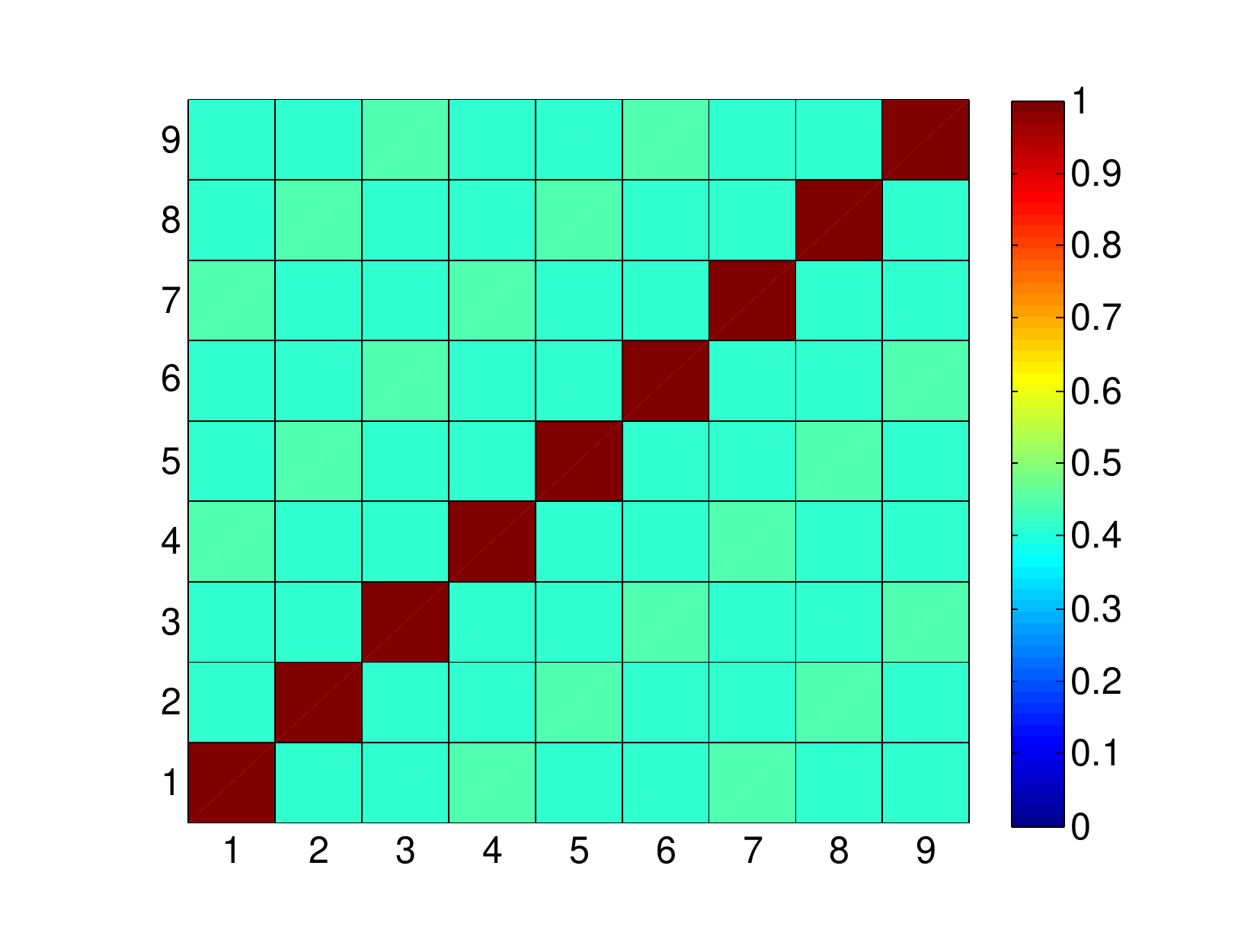}\hfill
\includegraphics[width=0.49\textwidth]{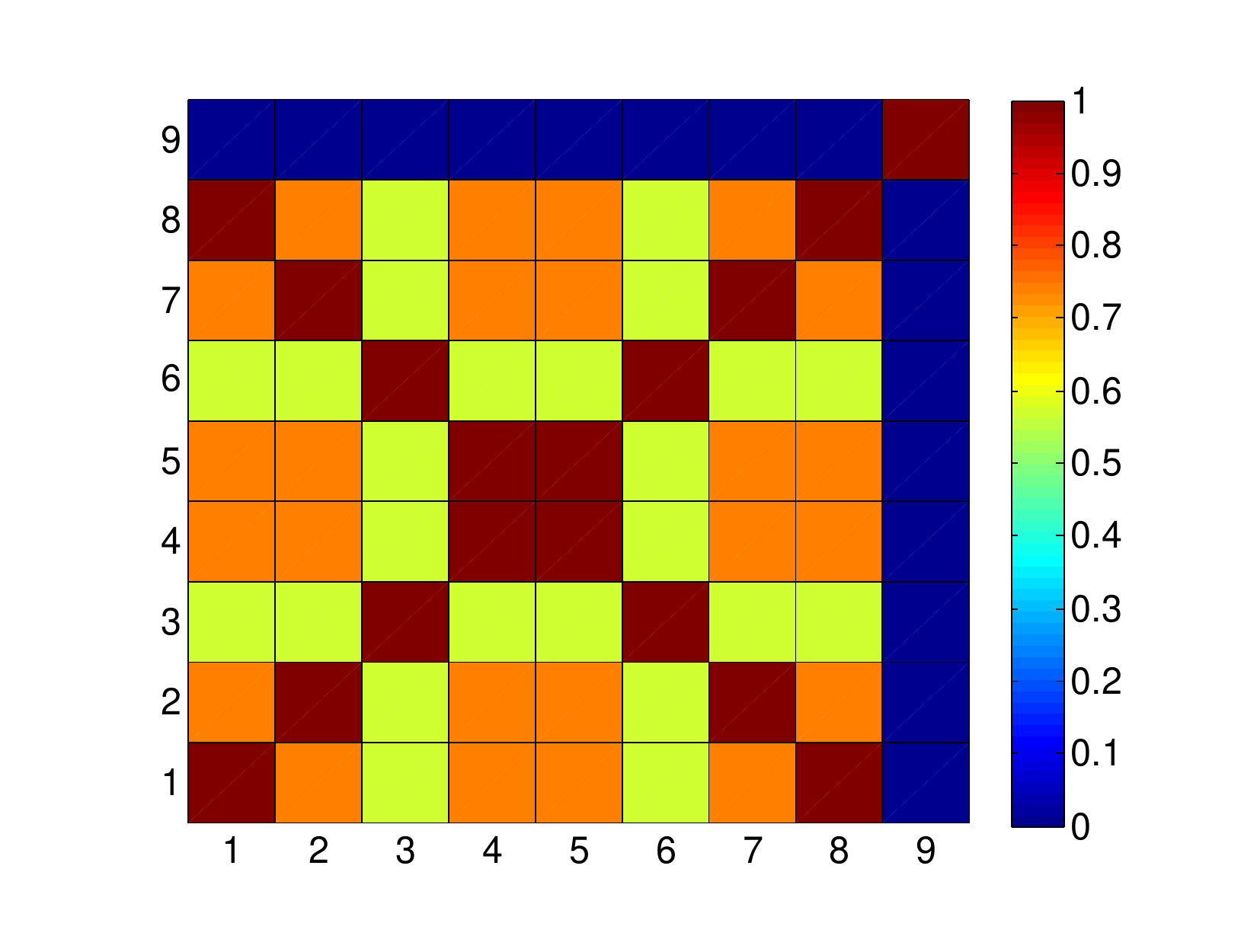}\\
\caption{Visual representation of maximum transfer fidelities for a
  ring of size 9 without bias (left) and the same ring with a finite
  bias applied to node 9 (right).}
\label{f:checker89}
\end{figure}

\begin{figure}[t]
  \includegraphics[width=0.49\textwidth]{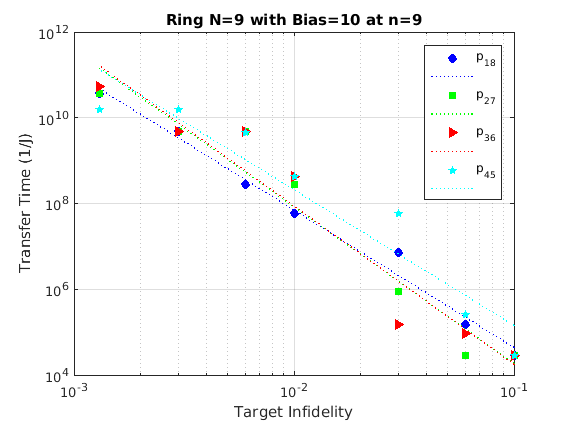}
  \includegraphics[width=0.49\textwidth]{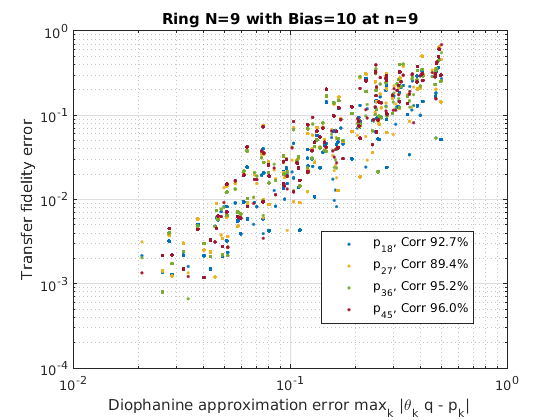}
  \caption{Scaling of transfer times for transfers between nodes $(1,8)$,
    $(2,7)$, $(3,6)$ and $(4,5)$ for a ring of size $N=9$ with a bias of strength
    $10$ (in units of $1/J$) applied to node $n=9$ (left) and correlation of
    fidelity error (or infidelity) and Diophantine approximation error (right).}
  \label{f:ring9-bias10-scaling}
\end{figure}

\section{Conclusion}

The concept of maximum transfer fidelity for information transfer
between nodes in a network of interacting spins was introduced and
criteria for attainability of the bounds in terms of the transition
frequencies of the network were given.  Attainability was shown to be
related, theoretically, to minimality of a linear flow and,
computationally, to a translation on a torus.  This last connection
enabled us to derive upper bounds on the time required to realize
transfer fidelities within $\epsilon_{\mathrm{prob}}$ of the maximum
transfer fidelity, for arbitrary $\epsilon_{\mathrm{prob}}>0$, via the
simultaneous Diophantine approximation.  Algorithms were discussed to
find the required approximations.

The ultimate aim of this analysis is to understand the intrinsic
limitations of information transfer in spin networks and utilize this
understanding to engineer networks with favorable bounds on the
information transfer fidelities and dynamic attainability properties,
so that high spin transfer fidelities can be attained in short times,
enabling fast transfer and minimizing the effects of noise and
decoherence.  An advantage of our approach of combining general ITF
bounds and asymptotic attainability conditions with an algorithm to
estimate the time required to achieve transfer within a set margin of
error, compared to engineering the spectrum of the network Hamiltonian
to admit perfect state transfer, for example, is that the latter
condition is generally too strong a requirement, as in practice there
are always margins of error.  Therefore it makes more sense to ask how
much time is required to achieve a certain transfer fidelity for a
given acceptable margin of error $\epsilon$, and try to optimize the
network topology, couplings or biases to achieve the best possible
transfer times for the acceptable margins of error.

The general results were applied specifically to regular spin
structures such as rings with uniform coupling.  In this case, the
information transfer infidelity prametric induced by maximum transfer
fidelity takes on full significance as it can be shown to be a proper
metric defining an information transfer infidelity geometry for the
network, which is significantly different from the physical network
geometry.  The analysis shows that the intrinsic transfer fidelities
for simple networks such as rings are often attainable asymptotically
but the times required to achieve high fidelities can be very long.
The intrinsic bounds on the ITFs and transfer times can be favorably
changed, however, by simple Hamiltonian engineering such as applying
spatially distributed static bias fields.  In particular, it was shown
how such simple controls can be used to alter the information transfer
fidelities and geometry of a network.  It was demonstrated how this
idea can be applied to enable or disable information transfer between
a pair of nodes in the network.  Simple bias controls are sufficient
to direct information flow between nodes.  By changing the biases
different transfers can be targeted, and thus a spin ring with fixed
couplings can be turned into a simple quantum router for information
encoded in excitations of a spin network.

Directions for future work include optimizing information transfer in
spin networks via optimal control to achieve faster and more efficient
dynamic routing in more complex spin networks.  While this work
focused on transfer of a single excitation, the concepts and analysis
can also be applied to the case of encoding and simultaneous transfer
of multiple excitations.  This is interesting as it could increase the
information transmission capacity of the network.  Finally, although
simulation results for similar spin systems suggest that some degree
of intrinsic robustness of state transfer and the ability to mitigate
the effects of noise, decoherence or fluctuations in the couplings via
control~\cite{Burgarth2007,Schirmer2009,Wang2010},
the sensitivity of transfer fidelities with regard to noise and
deleterious effect of the environment need to be investigated for
specific physical realizations of spin networks.

\section*{Acknowledgment}
E. A. Jonckheere was partially supported by the Army Research Office
(ARO) Multi University Research Initiative (MURI) grant
W911NF-11-1-0268.  S. G. Schirmer and F. C. Langbein acknowledge
support from the Ser Cymru National Research Network on Advanced
Engineering.  SGS also acknowledges funding from a Royal Society
Leverhulme Trust Senior Fellowship.


\begin{thebibliography}{10}
\providecommand{\url}[1]{#1}
\csname url@samestyle\endcsname
\providecommand{\newblock}{\relax}
\providecommand{\bibinfo}[2]{#2}
\providecommand{\BIBentrySTDinterwordspacing}{\spaceskip=0pt\relax}
\providecommand{\BIBentryALTinterwordstretchfactor}{4}
\providecommand{\BIBentryALTinterwordspacing}{\spaceskip=\fontdimen2\font plus
\BIBentryALTinterwordstretchfactor\fontdimen3\font minus
  \fontdimen4\font\relax}
\providecommand{\BIBforeignlanguage}[2]{{%
\expandafter\ifx\csname l@#1\endcsname\relax
\else
\language=\csname l@#1\endcsname
\fi
#2}}
\providecommand{\BIBdecl}{\relax}
\BIBdecl

\bibitem{Caruso2014}
Filippo Caruso, Vittorio Giovannetti, Cosmo Lupo, and Stefano Mancini
``Quantum channels and memory effects,''
\emph{Reviews of Modern Physics}, vol.~86, pp.1203-1259, 2014.

\bibitem{quantum_spintronics2013}
D.~D. Awschalom \emph{et~al.}, ``Quantum spintronics: Engineering and
  manipulating atom-like spins in semiconductors,'' \emph{Science}, vol. 339,
  no. 6124, pp. 1174--1179, March 2013.

\bibitem{Bose2003} 
S.~Bose, ``Quantum Communication through an Unmodulated Spin Chain,''
\emph{Phys. Rev. Lett.} vol. 91, no. 207901, 2003.

\bibitem{Bose2007} 
S. Bose, ``Quantum communication through spin chain dynamics: an
introductory overview,'' \textit{Contemp. Phys.} vol. 48, p.13-30, 2007.

\bibitem{Key2010} 
A.~Key, ``A Review of Perfect, Efficient, State Transfer and its
Application as a Constructive Tool,'' \emph{Int. J. Quantum Inform.},
vol 08, no. 641, 2010.

\bibitem{Christandl2004} 
M. Christandl, N. Datta, A. Ekert, and A. J. Landahl, ``Perfect State
Transfer in Quantum Spin Networks,''
\textit{Phys. Rev. Lett.} vol. 92 no. 187902, 2004.

\bibitem{Christandl2005}
M. Christandl, N. Datta, T. Dorlas, A. Ekert, A. Kay, and A. Landahl,
``Perfect transfer of arbitrary states in quantum spin networks,''
\textit{Phys. Rev. A} vol. 71, no. 032312, 2005.

\bibitem{Greentree2004}
A.D.~Greentree, J.H.~Cole, A.R.~Hamilton, and L.C.L.~Hollenberg, 
``Coherent electronic transfer in quantum dot systems using adiabatic passage,''
\textit{Phys. Rev. B} vol. 70, no. 235317, 2004.

\bibitem{Zueco2009}
David Zueco, Fernando Galve, Sigmund Kohler and Peter H\"anggi,
``Quantum router based on ac control of qubit chains,''
\emph{Phys. Rev. A} vol. 80, no. 042303, October 2009.

\bibitem{Schirmer2009}
  S.~G.~Schirmer and P.J.~Pemberton-Ross,
  ``Fast high-fidelity information transmission through spin-chain quantum wires,''
  \textit{Phys. Rev. A} vol. 80, no. 030301, 2009.

\bibitem{Wang2010}
  X.~Wang, A.~Bayat, S.~Bose, S.G~Schirmer, ``Global
  control methods for Greenberger-Horne-Zeilinger-state generation on
  a one-dimensional Ising chain,''
\textit{Phys. Rev. A} vol. 82, no. 012330, 2010.

\bibitem{Severini2009}
  Andrea Casaccino, Seth Lloyd, Stefano Mancini and Simone Severini,
 ``Quantum state transfer through a qubit network with energy shifts and fluctuations,''
  \emph{International Journal of Quantum Information}.
  vol.~7, no.~8, pp. 1417-1427, 2009.

\bibitem{MSC2011}  
E.~Jonckheere, S.~Schirmer, and F.~Langbein, ``Geometry and curvature of spin
  networks,'' in \emph{IEEE Multi-Conference on Systems and Control}, Denver,
  CO, September 2011, pp. 786--791, dOI: 10.1109/CCA.2011.6044395. Available at
  arXiv:1102.3208v1 [quant-ph].
  
\bibitem{QINP2014}
E.~Jonckheere, S.~Schirmer, and F.~Langbein, ``Quantum networks: {T}he
  anti-core of spin chains,'' \emph{Quantum Information Processing}, vol.~13,
  pp. 1607--1637, 2014. 

\bibitem{Wang2012}
  X.~Wang, P.~Pemberton-Ross, and S.~G. Schirmer, ``Symmetry \&
  controllability for spin networks with a single-node control,''
  \emph{IEEE Trans. Autom. Control}, vol.~57, pp.~1945--1956, 2012.

\bibitem{Massey2007} 
  A.~Massey, S.~J. Miller, and J.~Sinsheimer, ``Distribution of eigenvalues of
  real symmetric palindromic {T}oeplitz matrices and circulant matrices,''
  \emph{J. Theor. Probab.}, vol.~20, pp. 637--662, 2007.

\bibitem{Kaku1998}
M.~Kaku, \emph{Introduction to Superstrings}, ser. Graduate Texts in
Contemporary Physics. New York: Springer, 1998.

\bibitem{MIT2000}
N.~Prakash, \emph{Mathematical Perspectives on Theoretical Physics: A Journey
  from Black Holes to Superstrings}.  New Delhi: Tata McGraw-Hill, 2000.

\bibitem{KatokHasselblatt1997}
A.~Katok and B.~Hasselblatt, \emph{Introduction to the Modern Theory of
  Dynamical Systems}.  Cambridge, 1997.

\bibitem{Lagarias1982} 
J.~C. Lagarias, ``The computational complexity of simultaneous {D}iophantine
  approximation problems,'' in \emph{Foundations of Computer Science, 1982.
  SFCS '08. 23rd Annual Symposium on}, Nov 1982, pp. 32--39.

\bibitem{Kovacs2013} 
A.~Kovacs and N.~Tihanyi, ``Efficient computing of n-dimensional simultaneous
  {D}iophantine approximation problems,'' \emph{Acta Univ. Sapientiae,
  Informatica}, vol.~5, no.~1, pp. 16--34, 2013.

\bibitem{Hensley2005} 
  ``Simultaneous {D}iophantine approximation,'' April 2005, available
  at
  \texttt{http://www.math.tamu.edu/~Doug.Hensley/SimultaneousDiophantine.pdf}.

\bibitem{Nowak1984}
W.~G. Nowak, ``On simultaneous {D}iophantine approximation,'' \emph{Rendiconti
  del Circolo Matematico di Palermo}, vol. XXXIII, pp. 456--460, 1984.

\bibitem{Chevallier2011} 
  N.~Chevallier, ``A survey of best simultaneous {D}iophantine
  approximations,'' 2011, available at
    \texttt{http://www.math.uha.fr/chevallier/publication/meilleures3.pdf}.

  \bibitem{Bosma2012}  
W.~Bosma and I.~Smeets, ``Finding simultaneous {D}iophantine approximations
  with prescribed quality,'' in \emph{ANTS X: Tenth Algorithmic Number Theory
  Symposium 2012}, E.~W. Howe and K.~S. Kedlaya, Eds., UC San Diego, CA, July
  9-13 2012, preprint available at {\tt http://arxiv.org/pdf/1001.4455v1.pdf}.

  \bibitem{Hua1982}
H.~L. Keng, \emph{Introduction to Number Theory}.\hskip 1em plus 0.5em minus
  0.4em\relax Berlin, Heidelberg, New York: Springer-Verlag, 1982.

  \bibitem{Koblitz1988}
N.~Koblitz, \emph{A course in Number Theory and Cryptography}, ser. Graduate
  Texts in Mathematics.\hskip 1em plus 0.5em minus 0.4em\relax New York,
  Berlin, Heidelberg, London, Paris, Tokyo: Springer-Verlag, 1988.

  \bibitem{Lagarias1982b} 
  J.~C. Lagarias, ``Best simultaneous {D}iophantine approximations {II}.
  {B}ehavior of consecutive best approximations,'' \emph{Pacific Journal of
  Mathematics}, vol. 102, no.~1, pp. 61--88, 1982.

  \bibitem{Nowak1981} 
  W.~G. Nowak, ``A note on the simultaneous {D}iophantine approximation,''
  \emph{Manuscripta Math.}, vol.~36, pp. 33--46, 1981.

  \bibitem{Lagarias1982c}
  J.~C. Lagarias,  ``Best simultaneous {D}iophantine approximations. i. growth
  rates of best approximation denominators,'' \emph{Transactions of the
  American Mathematical Society}, vol. 272, no.~2, pp. 545--554, August 1982.

 \bibitem{Lenstra1982}
  K.~Lenstra, H.~W. ~Lenstra and L.~Lovasz, ``Factoring with rational coefficients,''
  \emph{Mathematische Annalen}~ vol.~261, no.~4, pp.515--534, 1982.
  
\bibitem{Lagarias1994} 
J.~C. Lagarias,, ``Geodesic multidimensional continued fractions,''  \emph{Proc. London
  Math. Soc. (3)}, vol.~69, pp. 464--488, 1994.

\bibitem{Pontryagin1990}
  A.~V. Arkhangel'skii and L.~S.~Pontryagin (Eds.), ``General topology i; basic concepts
  and constructions; dimension theory,'' in \emph{Encyclopedia of Mathematical
  Sciences}, R.~V. Gramkrelidze, Ed.\hskip 1em plus 0.5em minus 0.4em\relax
  Berlin, New York: Springer, 1990, vol.~17.

\bibitem{Aldrovandi1995} 
R.~Aldrovandi and J.~G. Pereira, \emph{An Introduction to Geometrical
  Physics}.\hskip 1em plus 0.5em minus 0.4em\relax Singapore, River Edge, NJ,
  London: World Scientific, 1995.

\bibitem{EURASIP2008}
F.~Ariaei, M.~Lou, E.~Jonckheere, B.~Krishnamachari, and M.~Zuniga, ``Curvature
  of indoor sensor network: clustering coefficient,'' \emph{EURASIP Journal on
  Wireless Communications and Networking}, vol. 2008, p. 20 pages, 2008,
  article ID 213185; doi: 10.1155/2008/2131185.

\bibitem{Wilson1931}
  W.~A. Wilson, ``On semi-metric spaces,''
  \emph{American Journal of Mathematics}, vol.~53, no.~2,
  pp. pp. 361--373, 1931.  Available:
  \url{http://www.jstor.org/stable/2370790}

\bibitem{Shore1981} 
  S.~D. Shore, ``Coherent distance functions,''
  \emph{Topology Proceedings}, vol.~6, pp. 405--422, 1981.

\bibitem{Blumenthal1943} 
  L.~M. Blumenthal, ``Some embedding theorems and characterization
  problems of distance geometry,'' \emph{Bull. Amer. Math. Soc.},
  vol.~40, pp. 321--338, 1943.

\bibitem{Jonckheere2011}
E.~Jonckheere, M.~Lou, F.~Bonahon, and Y.~Baryshnikov, ``Euclidean versus
  hyperbolic congestion in idealized versus experimental networks,''
  \emph{Internet Mathematics}, vol.~7, no.~1, pp. 1--27, March 2011.

\bibitem{ACC2014}
C.~Wang, E.~Jonckheere, and R.~Banirazi, ``Wireless network capacity versus
  {O}llivier-{R}icci curvature under {H}eat {D}iffusion ({HD}) protocol,'' in
  \emph{American Control Conference (ACC)}, Portland, OR, June 04-06 2014, pp.
  3536--3541, available at {\tt http://eudoxus2.usc.edu}.

\bibitem{Blumenthal1953}
L.~M. Blumenthal, \emph{Theory and Applications of Distance Geometry}.\hskip
  1em plus 0.5em minus 0.4em\relax London: Oxford at the Clarendon Press, 1953.

\bibitem{ISCCSP2012}
E.~Jonckheere, F.~C. Langbein, and S.~G. Schirmer, ``Curvature of quantum
  rings,'' in \emph{Proceedings of the 5th International Symposium on
  Communications, Control and Signal Processing (ISCCSP 2012)}, Rome, Italy,
  May 2-4 2012, {DOI}: 10.1109/{ISCCSP}.2012.6217863.

\bibitem{Gromov1987}
M.~Gromov, ``Hyperbolic groups,'' in \emph{Essays in Group Theory}, ser.
  Mathematical Sciences Research Institute Publication, S.~M. Gersten,
  Ed.\hskip 1em plus 0.5em minus 0.4em\relax New York: Springer-Verlag, 1987,
  vol.~8, pp. 75--263.
  
\bibitem{Burgarth2007}
D.~Burgarth, ``Quantum state transfer and time-dependent disorder in 
quantum chains,'' \emph{The European Physical Journal Special Topics}.
vol.~151, Issue 1, pp 147-155, December 2007.

\bibitem{CDC2015}
  F.~Langbein, S.~Schirmer, E.~Jonckheere,
  ``Time-optimal information transfer in spintronic networks,''
  to appear in \emph{Proc. of IEEE CDC 2015}, preprint available at
  arXiv: 1508:00928.
\end{thebibliography}
\end{document}